\documentclass[journal,letterpaper]{IEEEtran}
\usepackage{cite}
\usepackage{amsmath,amssymb,amsfonts, amsthm}
\usepackage{algorithmic}
\usepackage{algorithm}

\usepackage{array}
\usepackage[caption=false,font=normalsize,labelfont=sf,textfont=sf]{subfig}
\usepackage{textcomp}
\usepackage{stfloats}
\usepackage{url}
\usepackage{verbatim}
\usepackage{graphicx}
\usepackage{xcolor}
\usepackage{tabularx}
\usepackage{booktabs}
\usepackage{multirow}

\usepackage{bm}
\newtheorem{theorem}{Theorem}

\newtheorem{assumption}{Assumption}

\DeclareMathOperator*{\argmax}{arg\,max}
\DeclareMathOperator*{\argmin}{arg\,min}

\newcommand*{\dif}{\mathop{}\!\mathrm{d}}

\begin{document}

%\setlength{\textfloatsep}{3pt}
%\setlength{\dbltextfloatsep}{5pt}
%\setlength{\floatsep}{0pt}
%\setlength{\dblfloatsep}{0pt}
%\setlength{\intextsep}{0pt}

%\setlength{\abovedisplayskip}{2pt}
%\setlength{\belowdisplayskip}{2pt}
%\setlength{\abovedisplayshortskip}{2pt}
%\setlength{\belowdisplayshortskip}{2pt}

%\title{Energy Minimization in Wireless Powered Mobile Edge Computing Networks
%%\thanks{Identify applicable funding agency here. If none, delete this.}
%}
% \title{Joint User Association and Resource Allocation for Ultra-Dense Networks with Adaptive Semantic Communication}
\title{Joint User Association and Resource Allocation for Adaptive Semantic Communication in 5G and Beyond Networks}

\author{Xingqiu He,~\IEEEmembership{Member,~IEEE,}
  Chaoqun You,~\IEEEmembership{Member,~IEEE},
  Zihan Chen,~\IEEEmembership{Member,~IEEE},
  Yao Sun,~\IEEEmembership{Senior Member,~IEEE},
  Dongzhu Liu,~\IEEEmembership{Member,~IEEE},
  Tony Q. S. Quek,~\IEEEmembership{Fellow,~IEEE},
  Yue Gao,~\IEEEmembership{Fellow,~IEEE}
        % <-this % stops a space
% \thanks{This paper was produced by the IEEE Publication Technology Group. They are in Piscataway, NJ.}% <-this % stops a space
% \thanks{Manuscript received April 19, 2021; revised August 16, 2021.}
\thanks{X. He, C. You and Y. Gao are with the Institute of Space Internet, Fudan University, Shanghai,
China (emails: hexqiu@gmail.com, chaoqunyou@gmail.com, gao.yue@fudan.edu.cn).}
\thanks{Z. Chen is with the Singapore University of Technology and Design, Singapore 487372 (e-mail: zihan\_chen@mymail.sutd.edu.sg).}
\thanks{Y. Sun and D. Liu are with the University of Glasgow, G12 8QQ Glasgow, U.K. (e-mail: yao.sun@glasgow.ac.uk, dongzhu.liu@glasgow.ac.uk).}
% \thanks{D. Liu is with the School of Computing Science, University of Glasgow, U.K. (e-mail: dongzhu.liu@glasgow.ac.uk).}
\thanks{T. Q. S. Quek is with the Singapore University of Technology and Design,
Singapore 487372, and also with the Yonsei Frontier Lab, Yonsei University,
South Korea (e-mail: tonyquek@sutd.edu.sg).}
}

% The paper headers
\markboth{}{}

% \IEEEpubid{0000--0000/00\$00.00~\copyright~2021 IEEE}
% Remember, if you use this you must call \IEEEpubidadjcol in the second
% column for its text to clear the IEEEpubid mark.

% \author{\IEEEauthorblockN{Xingqiu He\IEEEauthorrefmark{1},
%         Chaoqun You\IEEEauthorrefmark{2},
%         Tony Q.S. Quek\IEEEauthorrefmark{3}}
%         \IEEEauthorblockA{Information Systems Technology and Design, Singapore University of Technology and Design \\
%         %Wherever\\
%         E-mail: \IEEEauthorrefmark{1}hexqiu@gmail.com,
%         \IEEEauthorrefmark{2}chaoqun\_you@sutd.edu.sg,
%         \IEEEauthorrefmark{3}tonyquek@sutd.edu.sg
%         }
%     }

\maketitle

\begin{abstract}
  Semantic communication (SemCom) has emerged as a promising paradigm that leverages Deep Neural Networks (DNNs) to extract task-relevant information, 
  thereby substantially reducing the volume of transmitted data. 
  In existing implementations, the semantic transceiver is typically pre-trained for a specific task and uniformly adopted by all users. 
  However, due to user heterogeneity in computational and communication capabilities, 
  employing a single, fixed semantic transceiver may degrade the coding efficiency and transmission robustness. 
  To address this issue, we first demonstrate the feasibility of dynamically adjusting the computational and communication overhead 
  of DNN-based semantic transceivers, enabling a more flexible paradigm referred to as Adaptive Semantic Communication (ASC).
  Building on this concept, we formulate a joint user association and resource allocation problem for ASC in 5G and beyond networks,
  aiming to maximize overall system utility under energy and latency constraints. 
  However, the problem is very challenging due to the inherent interdependencies among decision variables.
  To tackle this complexity, we decompose the original problem into three subproblems: (i) ASC scheme selection for each user, 
  (ii) spectrum allocation at each Small-cell Base Station (SBS), 
  and (iii) user association across SBSs. 
  Each subproblem is solved sequentially based on the solutions of the preceding stages. 
  The proposed algorithm efficiently yields near-optimal solutions with polynomial-time complexity. 
  Simulation results demonstrate our approach outperforms existing baselines under various situations.
\end{abstract}


\section{Introduction}
The rapid proliferation of data-intensive applications---such as virtual reality, autonomous systems, and smart cities---has resulted in an 
unprecedented surge in data demand \cite{ali2024metaverse, al2015iot}. 
Owing to the intrinsic limitations of communication resources (e.g., spectrum and time), such demand has already 
surpassed the transmission capacity of the conventional communication paradigm, 
which primarily emphasizes the accurate delivery of raw bits. 
To address this challenge, recent research has investigated innovative approaches that reduce the transmission volume 
through the integration of Artificial Intelligence (AI) techniques. 
In particular, a novel communication paradigm known as semantic communication (SemCom) has been proposed \cite{bao2011towards, qin2021semantic, gunduz2022beyond}. 
By leveraging AI, SemCom focuses on transmitting only the information that is meaningful and relevant to downstream tasks---commonly 
referred to as semantic information---thereby enabling substantial reductions in data traffic without compromising task performance.

Although SemCom effectively reduces the volume of transmitted data, 
it also introduces non-negligible computational overhead due to the extraction of semantic information.
In current implementations, users performing the same downstream task typically share a common SemCom transceiver.
However, due to the inherent heterogeneity in user devices and network conditions, it is desirable to dynamically adapt the transceiver 
based on the available computational and communication resources.
For instance, as semantic transceivers are commonly realized using Deep Neural Networks (DNNs), 
we can actively reduce the DNN model size to mitigate computational demands for users with constrained computational capabilities.
Similarly, for users experiencing poor channel conditions, the volume of transmitted data can be further decreased by selectively transmitting a subset of the semantic information. 
By jointly applying these two strategies, a spectrum of SemCom schemes can be constructed, 
each characterized by different levels of computational complexity, communication overhead, and inference accuracy.
Selecting the most appropriate scheme for each user leads to the concept of Adaptive Semantic Communication (ASC), 
which significantly enhances the efficiency of both communication and computaional resources.

To demonstrate the potential of ASC in future communication systems, this paper investigates the performance optimization of ASC in 5G and beyond networks,
where Small-cell Base Stations (SBSs) are densely deployed to meet the high-quality transmission demands of users.
The proposed optimization framework simultaneously addresses user association, resource allocation, and ASC scheme selection to maximize system utility
under energy and latency constraints.
However, this problem is inherently challenging due to the complex interdependencies among decision variables.
The primary challenges include:
\begin{itemize}
  \item \textbf{Computation-communication trade-off:}
  For each user, both extraction and transmission of semantic information consume time and energy. 
  With limited time and energy budgets, a fundamental trade-off exists between the computational and communication workloads that a user can handle. 
  To identify the optimal ASC scheme, it is crucial to characterize this trade-off and understand its effect on the inference accuracy of downstream tasks.

  \item \textbf{Spectrum sharing among heterogeneous users:}
  Users associated with the same SBS must share limited spectral resources. 
  However, due to user heterogeneity---such as variations in channel conditions, utility functions, and computational capabilities---the 
  marginal utility of spectrum differs across users, making optimal spectrum allocation a non-trivial problem.

  \item \textbf{User association under spatial and network heterogeneity:}
  In 5G and beyond networks, individual users are typically within the coverage of multiple SBSs. 
  However, associating each user with its nearest SBS can lead to significant load imbalance due to the uneven spatial distribution of users and SBSs. 
  Hence, optimizing user association requires a joint consideration of interrelated factors such as load balancing, channel conditions, and user-specific transmission demands.

  \item \textbf{Coupled decision variables across network layers:} 
  The decisions variables at different network levels are also deeply coupled. 
  For example, user association determines the grouping of users that share spectrum resources. 
  Furthermore, the spectrum allocation impacts each user’s transmission capacity, which influences the trade-off between communication and computation, 
  and ultimately affects the choice of ASC schemes.
\end{itemize}

To address the aforementioned challenges, we decompose the original optimization problem into three hierarchical subproblems,
each corresponding to a different network layer:

\begin{itemize}
\item \textbf{ASC Scheme Selection:} 
  We first focus on selecting the optimal ASC scheme for each user under fixed user association and spectrum allocation. 
  To characterize the computation-communication trade-off, we formally prove that both energy and latency constraints are tight at the optimal solution. 
  This key property enables us to efficiently determine the maximum transmittable data given a fixed computational workload.
  Building upon this result, we can derive the optimal computational workload and subsequently determine the optimal ASC scheme.

\item \textbf{Spectrum Allocation:} 
  Next, we address the spectrum allocation under fixed user association.
  We consider two classes of utility functions: concave and general. 
  For the concave case, we demonstrate that a simple greedy strategy yields the optimal allocation. 
  For general utility functions, we develop a dynamic programming-based approach to efficiently determine the optimal solution.

\item \textbf{User Association:} 
  Finally, we optimize user association across SBSs. 
  The proposed method iteratively updates each user's association starting from a dummy initialization. 
  To mitigate computational complexity, we employ a linear approximation to estimate the marginal utility change for each potential reassignment. 
  Furthermore, we introduce a kurtosis-based heuristic to determine the processing order of users.
\end{itemize}

These subproblems are inherently interrelated, with each depending on the optimal solutions of the preceding layers. 
By systematically leveraging the structural characteristics of each subproblem and their interconnections, 
we eventually obtain an efficient algorithm capable of producing near-optimal solutions with polynomial time complexity.
The main contributions of this paper are summarized as follows:
\begin{itemize}
\item We introduce an ASC model that characterizes each semantic scheme through its computation demand, communication demand, and task-level utility.
      Based on this model, we formulate a joint optimization problem for semantic scheme selection, RB allocation, and user association under latency and energy constraints.
      We further prove that the formulated problem is strongly NP-hard even under simplified network settings.

\item For the ASC scheme selection subproblem, we prove that both the latency and energy budgets are fully utilized at the optimum, 
      thereby reducing the optimal scheme selection problem to an efficient one-dimensional numerical search. 
      Building on this result, we further optimize RB allocation with explicit guarantees on optimality and computational complexity, 
      employing a greedy-based algorithm for concave utility functions and a dynamic programming algorithm for general utility functions.

\item For the user association subproblem, we propose a relaxation-and-refinement strategy that iteratively enhances the association decisions. 
      To avoid repeated full evaluations of candidate associations, the proposed method estimates potential gains using a linear approximation 
      of the marginal utilities. 
      Furthermore, to mitigate regret during sequential refinement, WDs with more uncertain SBS selections are prioritized according to a kurtosis-based ordering rule.

\item We conduct extensive simulations to evaluate the complete ASC optimization framework.
      The results compare the proposed method with multiple baselines and show consistent utility improvements under different utility functions, network scales, delay requirements, and energy budgets.
      The simulations also demonstrate the benefit of jointly adapting semantic schemes, RB allocation, and user association.
\end{itemize}

The rest of the paper is organized as follows.
In Section \ref{section:related_work}, we review related works.
In Section \ref{section:system_model}, we describe the system model and formulate our optimization problem.
Our algorithm design and theoretical analysis are presented in Section \ref{section:algorithm_design}.
In Section \ref{section:simulation}, we show the numerical results to validate the performance of our algorithm.
Section \ref{section:conclusion} concludes the paper and discusses open problems for future work.

\section{Related Work} \label{section:related_work}
\subsection{User Association in 5G and Beyond Networks}
To meet the continuous growth in data traffic, telecommunication operators have been compelled to deploy a large number of SBSs, 
giving rise to the concept of Ultra-Dense Networks (UDNs) and Heterogeneous Netoworks (HetNets) 
\cite{kamel2016ultra,teng2018resource,chen2016user,xu2021survey}.
Given the limited communication resources available at each SBS, the design of effective user association strategies is 
crucial for balancing network load across SBSs and enhancing overall system performance \cite{andrews2014overview, liu2016user}.

In \cite{bethanabhotla2015optimal}, the authors formulated a user association problem aimed at maximizing network utility by optimizing 
the fraction of time each user is associated with different SBSs. 
This work was later extended to cooperative SBS scenarios in \cite{ye2016user}. 
To address dynamic network conditions, a robust user association algorithm was proposed in \cite{8485896} based on predicted future traffic fluctuations. 
The long-term user association problem was studied in \cite{zhang2021learning}, 
where a matching game-based algorithm was developed to efficiently maximize data rate and minimize the number of handovers.

User association is also commonly studied in conjunction with resource allocation---such as power and bandwidth---to achieve improved system performance 
\cite{zhang2017energy, nguyen2020joint, hou2021user, si2021qos}. 
More recently, AI techniques have been introduced to solve the user association problem. 
Specifically, \cite{moon2023energy} proposed a decentralized multi-agent actor-critic framework to maximize energy efficiency in UDNs, 
while \cite{xue2021user} applied a deep Q-network to learn optimal user association strategies from historical experience. 
In \cite{lee2023deep}, a deep learning-based approach was developed to jointly optimize transmit power and user association.

Although user association is extensively investigated in existing literature, these studies are largely confined to the conventional communication paradigm. 
In contrast, our work introduces the ASC framework, which requires joint consideration of communication and computational resources,
as well as their impact on the inference accuracy of downstream tasks. 
This added complexity significantly differentiates our problem from traditional user association formulations,
thereby necessitating the development of novel algorithmic approaches.

\subsection{Resource Allocation for Semantic Communication}
The concept of SemCom was originally introduced decades ago \cite{shannon1949mathematical}, 
but practical implementation remained elusive until recent advances in AI technologies enabled new possibilities. 
In particular, the emergence of DNNs has spurred growing interest in leveraging AI to realize SemCom systems. 
In \cite{xie2021deep}, the authors proposed a pioneering SemCom framework for text transmission, 
which was subsequently extended to other modalities, 
including images \cite{huang2022toward}, speech \cite{weng2021semantic}, and video \cite{jiang2022wireless}. 
These implementations demonstrated substantial performance gains compared to conventional communication paradigms.

To further enhance SemCom systems, several studies have explored semantic-aware resource allocation. 
For example, \cite{yan2022resource} investigated spectral efficiency in the semantic domain and re-examined resource 
allocation from a semantic perspective. 
In \cite{yan2024qoe}, the authors introduced the concept of \emph{semantic entropy} to quantify task-relevant 
semantic information.
They also developed a novel Quality-of-Experience (QoE) model for semantic-aware resource optimization
and proposd a deep Q-network-based solution. 
Similarly, \cite{zhang2023optimization} designed an improved deep reinforcement learning (DRL) algorithm 
to enable server coordination during training, thereby accelerating convergence to a near-global optimum. 
% The work in \cite{wang2024adaptive} presented an adaptive semantic resource allocation framework 
% incorporating semantic-bit quantization, which is compatible with existing wireless communication systems. 
% This framework jointly optimizes transmit beamforming, semantic representation granularity, subchannel allocation, 
% and bandwidth distribution to maximize the overall effective quality of service (QoS).
The authors in \cite{xia2024joint} develop semantic channel models with a new performance metric
and systematically optimize user association and bandwidth allocation.
The work in \cite{cheng2025aigc} introduces a SemCom–empowered AIGC framework that leverages 
a resource-aware workload trade-off scheme to improve latency and content quality in dynamic wireless environments.

Recent research has also recognized that not all semantic information contributes equally to the inference 
accuracy of downstream tasks. 
As a result, prioritizing the transmission of high-impact semantic components has emerged as a promising strategy. 
For instance, \cite{wang2022performance} proposed a DRL-based algorithm using proximal policy optimization to jointly 
determine both resource allocation and the subset of semantic information to be transmitted. 
In \cite{zhang2023drl}, a dynamic resource allocation policy was designed using DRL to ensure that data with higher 
semantic value receives preferential access to limited communication resources. 
The study in \cite{liu2023adaptable} further explored the joint optimization of semantic information compression 
and resource block allocation to maximize task success probability.

Several recent works have further considered heterogeneous computational capabilities in multi-user SemCom systems.
In \cite{nguyen2024swin}, a Swin Transformer-based dynamic SemCom framework was proposed for multi-user image reconstruction, 
where the BS adapts the transmitted signal length according to channel conditions and users' computational capacities.
The work in \cite{nguyen2025adaptive} considered non-uniform deep architectures across users and designed a training framework 
to mitigate catastrophic forgetting in multi-user DeepJSCC.
More recently, \cite{nguyen2026anchor} proposed an anchor-aided training strategy with adaptive decoders to 
support users equipped with distinct decoder architectures.
Although these studies also consider heterogeneous computational capabilities, 
they primarily focus on semantic transceiver design, dynamic compression, or training stability under heterogeneous user architectures. 
In contrast, our work investigates network-level scheduling for ASC, in which user association, RB allocation, and semantic scheme selection are jointly optimized. 
Moreover, ASC introduces a strong coupling among computation workload, communication workload, and task performance, 
leading to a resource allocation problem that is structurally different from those in conventional networks and SemCom systems with fixed semantic transceivers.

% \textcolor{red}{
% While the aforementioned approaches effectively reduce communication overhead, 
% they still necessitate the generation of complete semantic information prior to transmission, 
% even if only a portion of it will be transmitted.
% The ASC framework proposed in this paper addresses this inefficiency by jointly optimizing the computational 
% and communication processes, ensuring that only the semantic information intended for transmission is generated. 
% However, this added flexibility also substantially increases the complexity of the underlying performance optimization problem.
% }

\section{System Model and Problem Formulation} \label{section:system_model}
In this section, we first present the considered system and describe the concept of ASC.
After that, we formulate an utility maximization problem and prove it is strongly NP-hard.
The major notations used in this paper are summarized in Table \ref{tab:notation}.

\begin{table}[!t]
% increase table row spacing, adjust to taste
\renewcommand{\arraystretch}{1.2}
 %if using array.sty, it might be a good idea to tweak the value of
 %\extrarowheight as needed to properly center the text within the cells
\caption{Major Notations}
\label{tab:notation}
\centering
% Some packages, such as MDW tools, offer better commands for making tables
% than the plain LaTeX2e tabular which is used here.
\begin{tabularx}{0.99\linewidth}{l l}
\hline
\textbf{Notation} & \textbf{Description}\\
\hline
$M$ & Number of SBSs. \\
$N$ & Number of WDs. \\
$A$ & Number of applications. \\
$x^m_n$ & Whether WD $n$ is associated with SBS $m$. \\
$z^m_n$ & Number of SBS $m$'s RBs allocated to WD $n$. \\
$c_n, d_n$ & Computation and communication workload of WD $n$. \\
$f_n$ & CPU frequency of WD $n$. \\
$P_n$ & Transmission power of WD $n$. \\
$u^a(c,d)$ & Utility function of application $a$. \\
$h^m_n$ & Channel gain between WD $n$ and SBS $m$. \\
$W$ & Spectrum bandwidth of each RB. \\
$K_m$ & Number of RBs available at SBS $m$. \\
$I_m$ & Average interference at SBS $m$. \\
$r^m_n$ & Transmission rate from WD $n$ to SBS $m$. \\
$T^c_n, T^t_n$ & Time consumption due to computation and transmission. \\
$E^c_n, E^t_n$ & Energy consumption due to computation and transmission. \\
\hline
\end{tabularx}
\end{table}

\subsection{Network Model}
\begin{figure}[t]
\centering
\includegraphics[width=0.4\textwidth]{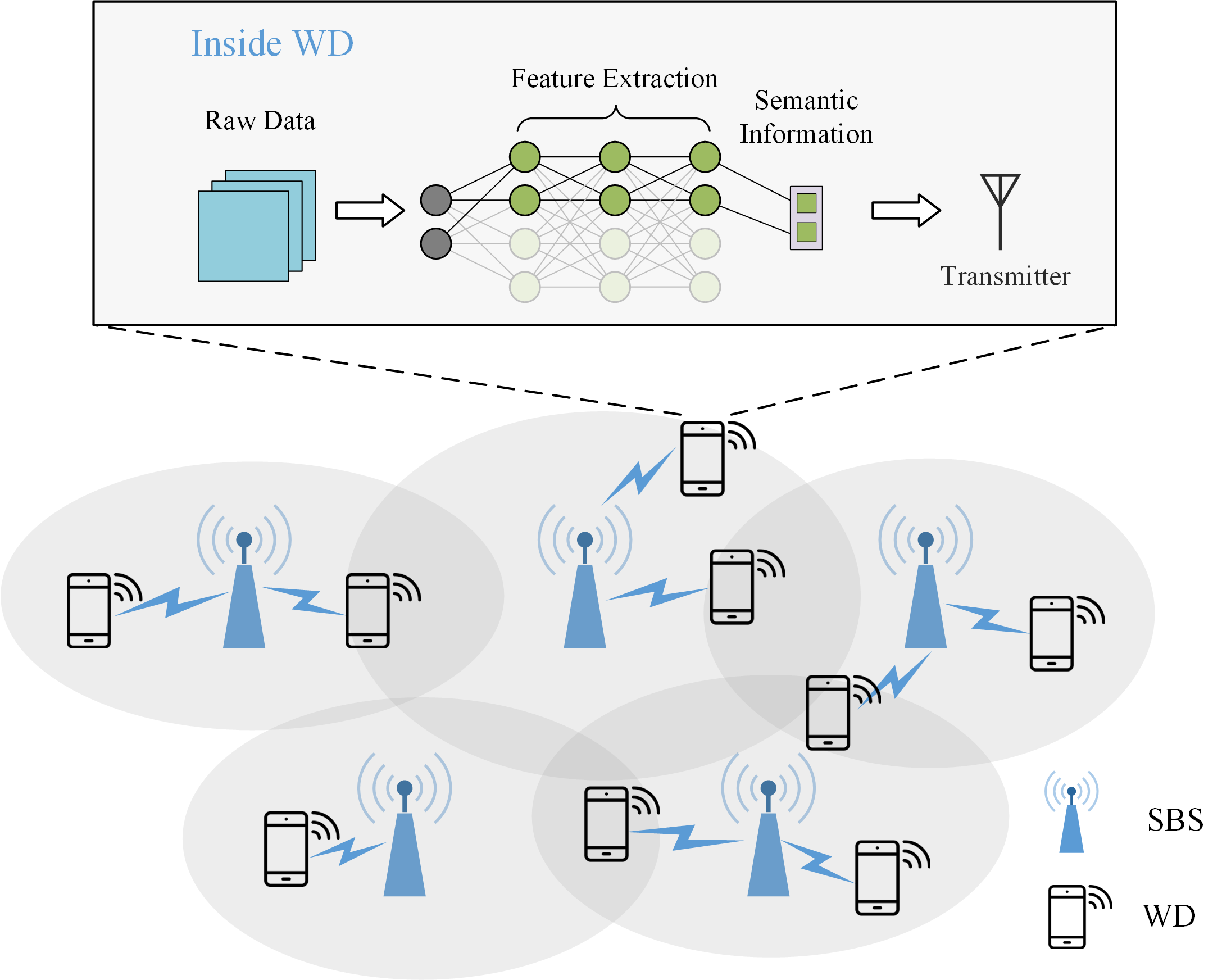}
\caption{A simple example of the considered system.}
\label{fig:system}
\end{figure}

As illustrated in Fig. \ref{fig:system}, we consider a 5G and beyond network comprising $M$ SBSs and $N$ wireless devices (WDs), 
where each WD may reside within the coverage area of multiple SBSs. 
The sets of SBSs and WDs are denoted by $\mathcal{M}$ and $\mathcal{N}$, respectively. 
The WDs may belong to various applications\footnote{In this paper, \emph{application} and \emph{downstream task} are used interchangeably.}, such as intelligent surveillance and smart transportation. 
Let $\mathcal{A}$ represent the set of $A$ applications deployed in the system, 
and $\mathcal{N}_a \subseteq \mathcal{N}$ denote the subset of WDs associated with application $a \in \mathcal{A}$.

Each WD in $\mathcal{N}_a$ generates raw data that must be transmitted to the SBSs for further processing. 
For instance, WDs may consist of surveillance cameras that upload captured images for security threat detection. 
Given the constraints of limited wireless resources, WDs perform semantic extraction locally to reduce data volume. 
Specifically, instead of transmitting raw data, WDs extract compact semantic representations (e.g., feature maps of the original images) and transmit only this semantic information to the SBSs.

In our model, each WD may reside within the coverage area of multiple SBSs, 
so we need to decide which SBS the WD should associate with for communication. 
The user association is represented by a binary variable $x^m_n$, where $x^m_n = 1$ indicates that WD $n$ is associated with SBS $m$, and $x^m_n = 0$ otherwise. 
We assume that each WD is associated with exactly one SBS. Accordingly, the user association variables must satisfy the following constraint:
\begin{equation}
  \sum_{m \in \mathcal{M}} x^m_n = 1, \quad \forall n \in \mathcal{N}.
  \label{cons:association}
\end{equation}

\subsection{Adaptive Semantic Communication Model} \label{subsection:schemes}

In SemCom systems, the desired semantic information varies across different downstream tasks. 
As a result, a DNN-based semantic transceiver has to be trained for each specific downstream task. 
Once trained, the transceiver is deployed to all WDs performing the corresponding task. 
However, due to the inherent heterogeneity in user capabilities, different WDs may prefer transceivers with varying computational and communication demands.

Consider two illustrative extreme cases: a WD with limited computational capacity but favorable channel conditions and abundant spectrum 
resources may find it preferable to transmit raw data directly to the SBS, bypassing local semantic extraction. 
In contrast, a WD with strong computational capability but constrained communication resources may choose to fully process the data 
locally and transmit only the final inference result. 
By selecting appropriate SemCom schemes tailored to device capabilities, 
the system can achieve more efficient utilization of both communication and computational resources, thereby improving overall performance.

In practical implementations, various techniques can be employed to generate multiple SemCom schemes. 
For instance, methods proposed in \cite{yu2018slimmable, yu2019universally, yu2019autoslim, Cai2020Once-for-All} enable training of DNNs 
with variable model sizes within a single training process. 
Smaller models reduce computational cost at the expense of lower inference accuracy. 
Given that semantic transceivers are implemented using DNNs, such techniques can be directly applied to produce a series of transceivers 
offering different trade-offs between efficiency and accuracy.

An alternative approach to constructing diverse SemCom schemes involves semantic compression 
\cite{wang2022performance, zhang2023drl, zhang2023optimization, liu2023adaptable}. 
This technique involves evaluating the significance of individual components within the extracted semantic information and selectively 
transmitting only those that contribute most to the accuracy of downstream inference. 
Although semantic compression introduces additional computational overhead for evaluating component importance, 
it can substantially reduce the amount of data to be transmitted.

As a concrete example, consider an image-classification SemCom system whose transmitter is a DNN encoder deployed at the WD.
Starting from a full base encoder, one may construct several subnetworks using width multipliers, e.g., $0.25$, $0.5$, $0.75$, and $1.0$, 
which correspond to different computational workloads.
For each subnetwork, the generated semantic feature channels can be ranked according to their contribution to validation accuracy, 
and only the top fraction of channels, e.g., $25\%$, $50\%$, $75\%$, or $100\%$, is quantized and transmitted.
Each pair of model width and transmitted feature fraction forms one ASC scheme $s_i=(c_i,d_i,u_i)$, 
where $c_i$ is the computation workload measured by the required CPU cycles or multiply–accumulate (MAC) operations, 
$d_i$ is the number of transmitted bits after quantization, and $u_i$ is the utility that depends on classification accuracy.
The set of ASC schemes is therefore generated by profiling these feasible model-width and feature-selection combinations.

In practice, the ASC scheme selection can be implemented through lightweight control signaling.
For each application, the feasible ASC schemes and their indices are generated offline and shared by the WD and the SBS during model deployment or capability registration.
During online operation, the proposed optimization can be performed by a logically centralized controller deployed within the local service area, for example at the macro BS or a C-RAN/SDN controller.
The controller collects device-capability, channel, and resource information from WDs and SBSs, and then determines the user association, RB allocation, and ASC scheme.
The controller then informs each WD of the selected scheme index and transmission parameters through downlink control signaling.

Based on the above discussion, for a given application $a$, 
we can define a utility function $u^a(c, d)$ to represent the utility derived from using a SemCom scheme with computation workload $c$ and 
communication workload $d$. 
It is important to note that in practice, the values of $c$ and $d$ are typically discrete. 
For example, in semantic compression, the semantic information is composed of discrete components that cannot be further subdivided, 
implying that $d$ must be an integer multiple of the size of each component. 
Nevertheless, since the feasible values of $c$ and $d$ are often densely distributed, 
the utility function $u^a(c, d)$ can still serve as a sufficiently accurate approximation in optimization formulations.
In Section \ref{subsection:scheduling}, we discuss how to select an implementable ASC scheme
based on the obtained computation and communication workload pair.

%Let $\mathcal{S}_a$ denote the set of $S_a$ available semantic communication schemes for application $a$.
%The $i$-th scheme in $\mathcal{S}_a$ can be expressed as $s_i = (c^a_i, d^a_i, u^a_i)$,
%Clearly, the result accuracy is increasing with respect to the computation workload $c$ and communication workload $d$.
%However, according to the results in 
%\cite{yu2018slimmable, yu2019universally, yu2019autoslim, Cai2020Once-for-All, 
%wang2022performance, zhang2023drl, zhang2023optimization, liu2022adaptable},
%the marginal return of $c$ and $d$ is diminishing.
%Therefore, we can assume the utility $u^a_i$ is ``concave'' with respect to $c^a_i$ and $d^a_i$.
%However, the concavity for discrete functions is not well-defined.
%To avoid this problem, in Section \ref{subsection:rb_allocation},
%we will give a more formal discussion on the concavity of $u^a_i$ and exploit this property in our algorithm design.

%we use binary variable $y^a_{ni} \in \{0,1\}$ to indicate whether WD $n \in \mathcal{N}_a$ chooses the $i$-th semantic communication scheme.
%Since each WD can choose at most one scheme, we have
%\begin{equation}
    %\sum_{i\in\mathcal{S}_a} y^a_{ni} \leq 1, \quad \forall n\in\mathcal{N}_a, \forall a\in\mathcal{A}.
    %\label{cons:scheme}
%\end{equation}
%Notice that we allow $\sum_{i\in\mathcal{S}_a} y^a_{ni} = 0$.
%In this case, the data at WD $n$ are not transmitted to the SBS thus no utility is obtained.

\subsection{Computation and Communication Model}
In current practice, semantic information extraction is typically performed using DNNs, which entail substantial computational workloads. 
To improve energy efficiency, we assume that WDs employ the dynamic voltage and frequency scaling technique \cite{weiser1996scheduling}, 
which enables dynamic adjustment of CPU frequency based on processing demands.
For WD $n\in \mathcal{N}_a$, let $f_n$ be its CPU frequency and $c_n$ be its computational workload, measured in CPU cycles.
Then the computation time required for semantic information extraction can be expressed as:
\begin{equation}
    T^c_n = \frac{c_n}{f_n}.
    \label{eq:computation_time}
\end{equation}
As established in \cite{chandrakasan1992low}, the energy consumption per CPU cycle is proportional to the square of the frequency. 
Accordingly, the total energy consumption for computation is given by:
\begin{equation}
    E^c_n = c_n \gamma_n f^2_n.
    \label{eq:computation_energy}
\end{equation}
where $\gamma_n$ is a hardware-dependent energy efficiency coefficient determined by the chip architecture.

It should be noted that the energy consumption model in \eqref{eq:computation_energy} corresponds to Dynamic Voltage and Frequency Scaling (DVFS)-enabled CPUs 
and has been widely adopted in mobile and edge computing studies. 
Nevertheless, our proposed algorithm is not limited to this specific CPU model and can be readily extended to more general 
energy consumption models denoted by $E^c_n=e_n(c_n,f_n)$.
In Section~\ref{subsection:scheduling}, we show how the proposed algorithm can be adapted to these energy consumption models.

On the communication side, we assume that the available spectrum at each SBS is equally divided into resource blocks (RBs) with a fixed bandwidth $W$.
The total number of RBs available at SBS $m$ is denoted by $K_m$.
Let $z^m_n\in \{0,1,\dots,K_m\}$ represent the number of SBS $m$'s RBs allocated to WD $n$.
In addition, the uplink interference at the SBS is typically the superposition of a large number of signals 
from spatially distributed users. 
Due to high user density and effective interference mitigation techniques (e.g., SIC, ZF, MMSE), 
the impact of any single user's transmit power on the total interference is often negligible. 
Therefore, to facilitate analytical tractability, we assume that uplink interference from other cells remains static, 
consistent with prior studies \cite{kim2011distributed, saad2014college, analytical2016, learning2015}.
However, this approximation may become less accurate when 
the interference is dominated by a few strong neighboring cells or when adjacent-cell traffic is highly bursty.
In these scenarios, a more refined interference model would be required, which is left for future work.

Under this assumption, the achievable transmission rate from WD $n$ to SBS $m$ is given by the Shannon capacity formula:
\begin{equation}
    r^m_n = z^m_n W \log_2 \left( 1 + \frac{P_n h^m_n}{\sigma^2 + I_m} \right),
    \label{eq:transmission_rate}
\end{equation}
where $P_n$ is the transmission power of WD $n$, $h^m_n$ denotes the channel gain between WD $n$ and SBS $m$,
$\sigma^2$ is the noise power at SBS $m$, and $I_m$ represents the average interference experienced by SBS $m$.
To ensure resource feasibility, the total number of RBs allocated by each SBS must not exceed its capacity:
\begin{equation}
    \sum_{n\in\mathcal{N}} z^m_n \leq K_m, \quad \forall m\in\mathcal{M}.
    \label{cons:RB}
\end{equation}
For WD $n\in\mathcal{N}_a$, let $d_n$ be the amount of semantic data to be transmitted, then the transmission time can be calculated as:
\begin{equation}
    T^t_n = \frac{d_n}{\sum_{m\in\mathcal{M}}x^m_n r^m_n}
    \label{eq:transmission_time}
\end{equation}
and the corresponding energy consumption is:
\begin{equation}
    E^t_n = P_n T^t_n.
    \label{eq:transmission_energy}
\end{equation}

\subsection{Problem Formulation}
In this paper, we aim to maximize the overall system utility by jointly optimizing user association, resource allocation, and SemCom scheme, 
subject to constraints on energy consumption, spectrum availability, and delay. 
Based on the system models developed above, the corresponding optimization problem can be formulated as follows:
\begin{align}
    \max_{\substack{\bm{x}, \bm{z}, \bm{f}, \bm{P} \\ \bm{c}, \bm{d}}}\quad & \sum_{a\in\mathcal{A}} \sum_{n\in\mathcal{N}_a} u^a(c_n, d_n) \label{problem} \\
    s.t.\quad & \eqref{cons:association} \mbox{ and } \eqref{cons:RB}  \quad \tag{\ref{problem}{a}} \label{cons:allocation} \\
    & T^c_n + T^t_n \leq T^{max}, \quad \forall n\in\mathcal{N} \tag{\ref{problem}{b}} \label{cons:delay} \\
    & E^c_n + E^t_n \leq E^{max}_n, \quad \forall n\in\mathcal{N} \tag{\ref{problem}{c}} \label{cons:energy}\\
    & 0 \leq f_n \leq f^{max}_n, \quad \forall n\in\mathcal{N} \tag{\ref{problem}{d}} \label{cons:f}\\
    & 0 \leq P_n \leq P^{max}_n, \quad \forall n\in\mathcal{N} \tag{\ref{problem}{e}} \label{cons:p}\\
    & x^m_n \in \{0,1\}, \quad \forall m\in\mathcal{M}, \forall n\in\mathcal{N} \tag{\ref{problem}{f}} \label{cons:x} \\
    & z^m_n \in \{0,1,\dots,K_m\}, \quad \forall m\in\mathcal{M}, \forall n\in\mathcal{N} \tag{\ref{problem}{g}} \label{cons:z}
\end{align}
Here, $\bm{x}, \bm{z}, \bm{f}, \bm{P}, \bm{c}, \bm{d}$ are the vectors of decision variables.
Constraint \eqref{cons:allocation} contains the user association and RB allocation constraints described above.
Constraint \eqref{cons:delay} ensures that the total latency, comprising both computation and transmission time, does not exceed the maximum delay threshold $T^{max}$.
Constraint \eqref{cons:energy} limits the total energy consumption of each WD to a specified maximum $E^{max}_n$.
Constraints \eqref{cons:f} and \eqref{cons:p} define the feasible ranges for CPU frequency and transmission power, bounded by their respective hardware limits 
$f^{max}_n$ and $P^{max}_n$.
Finally, constraints \eqref{cons:x} and \eqref{cons:z} define the domain of binary user association variables and integer-valued RB allocation variables.

The formulated optimization problem is a mixed-integer programming (MIP) problem, which is generally challenging to solve.
In the literature, NP-hardness is commonly employed to characterize the inherent difficulty of such problems. 
However, establishing NP-hardness alone does not necessarily imply that a problem is intractable in practice. 
In particular, \emph{weakly} NP-hard problems---such as the classical 0-1 knapsack problem---can be addressed using pseudo-polynomial time algorithms, 
and are often solved efficiently through dynamic programming \cite{martello1990knapsack}. 
In contrast, we demonstrate that the problem under consideration is \emph{strongly} NP-hard, 
which indicates it admits no pseudo-polynomial time algorithm and is fundamentally more complex to solve.

\begin{theorem}
    Problem \eqref{problem} is strongly NP-hard even if 
    (i) there is only one WD for each running application,
    (ii) the energy budgets $E^{max}_n$ are sufficiently large, and 
    (iii) the channel gains $h^m_n$ and interference levels $I_m$ are identical across all SBSs.
    \label{theorem:np}
\end{theorem}
\begin{IEEEproof}
    We consider a specific class of two-valued utility functions, defined as $u^a(c,d) = u^a_0$ for all $(c,d) \geq (c^a_0,d^a_0)$
    and $u^a(c,d) = 0$ otherwise, where $u^a_0, c^a_0, d^a_0$ are application-specific positive constants.
    According to condition (i), when WD $n$ is specified, the corresponding application $a$ is also determined.
    For notational convenience, we omit the superscript $a$ and use $u_n$ to denote $u^a_0$.
    Since the utility function is two-valued, WD $n$'s utility is either $0$ or $u_n$.
    Let $y^m_n\in\{0,1\}$ indicate whether WD $n$ is served by SBS $m$ with positive utility, where $\sum_{m\in\mathcal{M}}y^m_n\leq 1$ because each WD can obtain the positive utility through at most one associated SBS.
    If WD $n$ obtains zero utility, it can still be associated with an arbitrary SBS due to \eqref{cons:association}, but we can set $c_n=d_n=z^m_n=0$ so that it consumes no RBs and contributes zero utility.
    Therefore, without loss of generality, a WD contributes utility $u_n$ only when $c_n = c^a_0$, $d_n = d^a_0$, and $y^m_n=1$ for some SBS $m$.
    Then the objective function of problem \eqref{problem} can be re-written as
    $\sum_{n\in\mathcal{N}} \sum_{m\in\mathcal{M}} y^m_n u_n$.

    According to condition (ii), the energy constraint \eqref{cons:energy} is redundant, allowing us to fix the CPU frequency and transmission 
    power at their maximum values, i.e., $f_n = f^{max}_n$ and $P_n = P^{max}_n$.
    Since both $f_n$ and $c_n$ are fixed, the computation time $T^c_n$ becomes a constant.
    Therefore, the delay constraint \eqref{cons:delay} is equivalent to $T^t_n \leq T^{max}-T^c_n$,
    where the right-hand side is a constant.
    By substituting the expression for $T^t_n$ from \eqref{eq:transmission_time}, 
    this constraint can be interpreted as a lower bound on the number of RBs allocated to WD $n$.
    Consequently, combining with condition (iii), the delay constraint \eqref{cons:delay} can be equivalently expressed as:
    $z^m_n \geq y^m_n C_n$, where $C_n$ is a constant independent of $m$.
    Combining with constraints \eqref{cons:RB} and \eqref{cons:z}, we derive an equivalent constraint on $y^m_n$:
    \begin{equation*}
        \sum_{n\in\mathcal{N}} y^m_n \lceil C_n \rceil \leq K_m, \quad \forall m\in\mathcal{M},
    \end{equation*}
    where $\lceil \cdot \rceil$ is the ceiling function that returns the smallest integer which is greater than or equal to the input value.
    As a result, the original problem \eqref{problem} reduces to a classical $0$-$1$ multiple knapsack problem, 
    which is known to be strongly NP-hard \cite{martello1990knapsack}.
    % Please see Appendix \ref{appendix:1}.
\end{IEEEproof}

Theorem \ref{theorem:np} highlights the intrinsic complexity of problem \eqref{problem}, 
showing that it remains computationally intractable even in highly simplified settings. 
Consequently, it is generally infeasible to obtain the exact optimal solution for non-trivial network scale.
To address this issue, we propose an efficient algorithm capable of producing near-optimal solutions in polynomial time.

\section{Algorithm Design and Performance Analysis} \label{section:algorithm_design}
In this section, we solve the formulated optimization problem through a three-stage framework. 
First, we determine the optimal scheduling variables $\bm{f}, \bm{P}, \bm{c}, \bm{d}$ assuming fixed user association and RB allocation. 
Building on this result, we can optimize the RB allocation for any given user association. 
Finally, we derive the optimal user association by applying a relaxation-and-refinement strategy.

\subsection{Scheduling Subproblem for Each WD} \label{subsection:scheduling}
Given fixed user association $\bm{x}$ and RB allocation $\bm{z}$, the original problem \eqref{problem} can be decomposed into $N$ independent subproblems,
each corresponding to the scheduling decision for a single WD.
Specifically, for any WD $n\in\mathcal{N}_a$ associated with SBS $m$, its local scheduling subproblem can be formulated as:
\begin{align}
    \max_{{f_n}, {P_n}, c_n, d_n}\quad & u^a(c_n, d_n) \label{subproblem} \\
    s.t.\quad & T^c_n + T^t_n \leq T^{max}, \tag{\ref{subproblem}{a}} \label{sub:cons:delay} \\
    & E^c_n + E^t_n \leq E^{max}_n, \tag{\ref{subproblem}{b}} \label{sub:cons:energy}\\
    & 0 \leq f_n \leq f^{max}_n, 0 \leq P_n \leq P^{max}_n \tag{\ref{subproblem}{c}} \label{sub:cons:fp}
\end{align}
We define a pair $(c_n,d_n)$ as achievable if there exists $f_n$ and $P_n$
such that the tuple $(c_n, d_n, f_n, P_n)$ satisfies constraints \eqref{sub:cons:delay}-\eqref{sub:cons:fp}.
The set of all achievable workload pairs constitutes the \emph{achievable region}.
Based on our discussion in Section \ref{subsection:schemes}, we can assume
the utility function $u^a(c_n, d_n)$ is non-decreasing, hence the optimal pair $(c^*_n, d^*_n)$ must lie on the
boundary of the achievable region, i.e. there is no achievable pair $(c_n, d_n)$ such that $c_n = c^*_n$ and $d_n > d^*_n$.
To characterize this boundary, it is sufficient to compute the maximum achievable $d_n$ for any given $c_n$, denoted by $d^{max}_n(c_n)$.
This process leads to the following optimization problem:
\begin{align}
    \max_{f_n, P_n}\quad & d_n \label{prob:d_max} \\
    s.t.\quad & \eqref{sub:cons:delay}-\eqref{sub:cons:fp}. \notag
\end{align}
To solve problem \eqref{prob:d_max}, we first demonstrate that its optimal solution satisfies the following tightness condition:
\begin{theorem}
    The optimal solution $(f^*_n, P^*_n)$ of problem \eqref{prob:d_max} must satisfy
    \begin{align}
        T^c_n + T^t_n &= T^{max} \label{eq:delay} \\
        E^c_n + E^t_n &= E^{max}_n. \label{eq:energy}
    \end{align}
    \label{prop:tightness}
\end{theorem}
\begin{IEEEproof}
    We prove by contradiction. We first show that \eqref{eq:energy} must be satisfied.
    If $E^c_n + E^t_n < E^{max}_n$, then we can assign the extra energy to the data transmission process while keeping the transmission duration unchanged.
    Specifically, let 
    \begin{equation*}
      P_n' = P_n^* + \frac{E^{max}_n - E^c_n - E^t_n}{T^t_n}.
    \end{equation*}
    Since $P'_n>P^*_n$, the achievable transmission rate increases from $r^m_n(P^*_n)$ to $r^m_n(P'_n)$.
    We then set $d'_n=T^t_n r^m_n(P'_n)$.
    By construction, the new transmission time is $T^{t,\prime}_n=d'_n/r^m_n(P'_n)=T^t_n$, whereas $d'_n>d_n$.
    Thus, $(f^*_n, P_n')$ is also feasible but produces larger $d_n$ than $(f^*_n, P^*_n)$, contradicting the fact that $(f^*_n, P^*_n)$ is optimal.

    Next, we show that \eqref{eq:delay} must be satisfied.
    If $T^c_n + T^t_n < T^{max}$, then we can assign the extra time to the computation process,
    i.e., let the new computation time be $T^{c,\prime}_n = T^{max} - T^t_n$.
    To process $c_n$ workload in $T^{c,\prime}_n$, the corresponding CPU frequency is
    \begin{equation*}
    f_n' = f_n^* - \frac{(T^{max} - T^t_n - T^c_n)c_n}{T^c_n(T^{max} - T^t_n)}.
    \end{equation*}
    Since $f_n' < f_n^*$, we can save an amount of $\Delta E^c_n = c_n\gamma_n [(f_n^*)^2 - (f_n')^2]$
    energy for the computation process.
    Then the saved energy can be assigned to the communication process, i.e., $P_n' = P^*_n + \Delta E^c_n / T^t_n$.
    We again set $d'_n=T^t_n r^m_n(P'_n)$, so the transmission time remains $T^t_n$ by construction while $d'_n>d_n$.
    One can easily verify that $(f_n', P_n')$ is feasible and produces larger $d_n$ than $(f_n^*, P_n^*)$,
    which results in a contradiction.
    % Please see Appendix \ref{appendix:2}.
\end{IEEEproof}

With Theorem \ref{prop:tightness}, we can replace the inequality constraints \eqref{sub:cons:delay} and \eqref{sub:cons:energy} in problem \eqref{prob:d_max}
with equalities \eqref{eq:delay} and \eqref{eq:energy}.
Substituting \eqref{eq:computation_time} and \eqref{eq:transmission_time} into \eqref{eq:delay} 
and using the fact that WD $n$ is associated with SBS $m$ (i.e. $x^m_n = 1$) result in
\begin{equation}
    P_n = \left(2^{\frac{d_n f_n}{(f_n T^{max} - c_n) z^m_n W}}-1\right) \frac{\sigma^2+I_m}{h^m_n}.
    \label{eq:P_n}
\end{equation}
Equation \eqref{eq:P_n} implicitly requires $f_nT^{max}-c_n>0$.
Later we will derive the feasible region of $c_n$ that guarantees there exists a feasible $f_n$ that satisfies this constraint.
Combining \eqref{eq:P_n} with the transmit-power upper bound in \eqref{sub:cons:fp}, we obtain the following feasibility condition:
\begin{equation}
    \frac{z^m_n W c_n \log_2 \left( 1 + \frac{P_n^{max} h^m_n}{\sigma^2+I_m} \right)}{z^m_n W T^{max} \log_2 \left( 1 + \frac{P_n^{max} h^m_n}{\sigma^2+I_m} \right) - d_n}
    \leq f_n \leq f_n^{max}.
    \label{feasible_region_f}
\end{equation}
The dependence on $d_n$ in \eqref{feasible_region_f} is removed after applying the tight energy constraint below.
Specifically, substituting \eqref{eq:computation_energy}, \eqref{eq:transmission_energy} and \eqref{eq:P_n} into \eqref{eq:energy} yields
\begin{align}
    d_n = &z^m_n W \left( T^{max} - \frac{c_n}{f_n} \right) \times \notag \\
    &\quad \log_2 \left( 1 + \frac{h^m_n (E^{max} - c_n \gamma_n f^2_n)}{(T^{max} - \frac{c_n}{f_n}) (\sigma^2+I_m)}\right). \label{eq:d}
\end{align}
For a general computation-energy model $E^c_n=e_n(c_n,f_n)$, 
the same procedure can be applied by replacing $c_n\gamma_n f_n^2$ in \eqref{eq:d} with $e_n(c_n,f_n)$.
Therefore, the closed-form expression in \eqref{eq:d} is specific to the DVFS model, 
whereas the decomposition framework and upper-layer allocation algorithms are compatible with other monotone energy consumption models.

For a given $c_n$, 
substituting \eqref{eq:d} into \eqref{feasible_region_f} gives the feasible region of $f_n$.
Consequently, solving problem \eqref{prob:d_max} reduces to finding the CPU frequency $f_n$ that maximizes \eqref{eq:d} over this feasible region.
This can be achieved by solving $\frac{\partial d_n}{\partial f_n} = 0$ and comparing the values of $d_n$ at the feasible critical and boundary points.
Although the analytical form of the stationary point (where $\frac{\partial d_n}{\partial f_n} = 0$) is intractable due to the complexity of the expression,
the value of $f^*_n$ can be efficiently computed using numerical techniques.\footnote{
In our implementation, we employ the \texttt{fzero} function in MATLAB, which 
uses Brent's method \cite{brent1971algorithm}, a hybrid root-finding approach combining bisection, secant and inverse quadratic interpolation methods.}

The preceding analysis establishes the boundary function $d^{\max}_n(c_n)$, 
which gives the maximum transmittable semantic data for each computation workload $c_n$. 
We now use this boundary characterization to determine the optimal workload. 
For notational convenience, we use $f^*_n(c_n)$ instead of $f^*_n$ in the following to explicitly express the dependence of the optimal CPU frequency 
on the computation workload.
According to \eqref{eq:d}, $d_n$ can be regarded as a function of $c_n$ and $f_n$, denoted by $d_n(c_n,f_n)$.
Thus we have
\begin{equation*} 
  d^{max}_n(c_n) = \max_{f_n} d_n(c_n, f_n) = d_n(c_n, f^*_n(c_n)).
\end{equation*}
Since the utility function is non-decreasing, the optimal workload pair must lie on the upper boundary of the achievable region.
Recall that this boundary is characterized by $(c_n,d^{max}_n(c_n))$.
Therefore, problem \eqref{subproblem} is reduced to the scalar optimization problem $\max_{c_n} u^a(c_n,d^{max}_n(c_n))$.
Similar to solving $f^*_n(c_n)$, the optimal $c^*_n$ can be obtained by comparing 
the stationary points of this scalar problem, which are determined by solving
\begin{align*}
    &\frac{\dif u^a(c_n, d^{max}_n(c_n))}{\dif c_n} \\
    &\quad = \frac{\partial u^a}{\partial c_n} + \frac{\partial u^a}{\partial d^{max}_n} \times \frac{\dif d^{max}_n}{\dif c_n} \\
    &\quad = \frac{\partial u^a}{\partial c_n} + \frac{\partial u^a}{\partial d^{max}_n} \times 
    \left( \frac{\partial d_n}{\partial c_n} + \frac{\partial d_n}{\partial f^*_n}\times \frac{\dif f^*_n}{\dif c_n} \right) \\
    &\quad = \frac{\partial u^a}{\partial c_n} + \frac{\partial u^a}{\partial d^{max}_n} \times \frac{\partial d_n}{\partial c_n} = 0,
\end{align*}
where the second equality holds because $d^{max}_n(c_n) = d_n(c_n, f^*_n(c_n))$
and the third equality holds because $\frac{\partial d_n}{\partial f^*_n}=0$ as $f^*_n$ is a stationary point of $d_n$.
From the latency and energy constraints, namely \eqref{sub:cons:delay} and \eqref{sub:cons:energy}, the feasible region of $c_n$ can be derived as
\begin{equation*}
  0 \leq c_n \leq \min\left( f^{max}_n T^{max}, \sqrt[3]{\frac{E^{max}_n (T^{max})^2}{\gamma_n}}\right).
\end{equation*}
It can be easily verified that there exists a feasible $f^*_n(c_n)$ for each $c_n$ in this region.
To optimize $c_n$, we numerically find the stationary points of $u^a(c_n,d^{max}_n(c_n))$ in this interval.
These stationary points, together with the two boundary points of the interval, constitute the set of candidate optimal computation workloads. 
For each candidate $c_n$, we compute the corresponding $d^{max}_n(c_n)$ and the resulting utility $u^a(c_n, d^{max}_n(c_n))$.
The optimal computation workload $c^*_n$ is then obtained by comparing these utility values.
Once $c^*_n$ is determined, the remaining scheduling decisions $f^*_n$, $d^*_n$, and $P^*_n$ can be readily computed from 
$f^*_n(c^*_n)$, \eqref{eq:d}, and \eqref{eq:P_n}, respectively.

In practical deployment, we need to obtain an implementable ASC scheme from the obtained continuous workload pair $(c^*_n,d^*_n)$.
Given the finite ASC scheme library $\mathcal{S}_a=\{s^a_i=(c^a_i,d^a_i,u^a_i)\}$,
WD $n$ selects the feasible scheme with the largest profiled utility under the allocated resources:
\begin{equation*}
    s^{impl}_n =
    \argmax_{\substack{s^a_i=(c^a_i,d^a_i,u^a_i)\in\mathcal{S}_a\\
    (c^a_i,d^a_i)\text{ is achievable under the allocated resources}}}
    u^a_i .
\end{equation*}
When $\mathcal{S}_a$ is dense, this selection can be regarded as projecting $(c^*_n, d^*_n)$ to a nearby feasible scheme with negligible utility loss.
When $\mathcal{S}_a$ is sparse, the per-WD scheduling step can also be solved directly by enumerating all schemes in $\mathcal{S}_a$.
The RB allocation and user association stages remain unchanged because they only require the resulting value $u^*_n(z^m_n)$ for each possible RB allocation.

\subsection{RB Allocation Subproblem for Each SBS} \label{subsection:rb_allocation}
In the previous subsection, we derived the optimal scheduling decisions $\bm{f}, \bm{P}, \bm{c}, \bm{d}$ under given user association $\bm{x}$ and 
RB allocation $\bm{z}$.
In this subsection, we continue to assume that $\bm{x}$ is fixed and focus on determining the optimal RB allocation $\bm{z}$.
Let $u^*_n(z^m_n)$ denote the maximum utility achieved by WD $n$ (i.e., the optimal objective value of problem \eqref{subproblem})
when it is allocated $z^m_n$ RBs.
Given the user association $\bm{x}$, we define $\mathcal{N}_m = \{ n\in\mathcal{N} \ | \ x^m_n = 1 \}$ as the set of WDs associated with SBS $m$.
Then the RB allocation subproblem for SBS $m$ can be formulated as:
\begin{align}
    \max_{z^m_n}\quad & \sum_{n\in\mathcal{N}_m} u^*_n(z^m_n) \label{subproblem:allocation} \\
    s.t.\quad & \sum_{n\in\mathcal{N}_m} z^m_n \leq K_m, \tag{\ref{subproblem:allocation}a} \label{cons:subproblem:rb_allocation} \\
    &z^m_n \in \{0,1,\dots,K_m\}, \quad \forall n\in\mathcal{N}_m \tag{\ref{subproblem:allocation}b} \label{cons:subproblem:z}.
\end{align}

In this subsection, we will present two algorithms for problem \eqref{subproblem:allocation} tailored to different classes of utility function $u^a(c,d)$.
In many practical applications, such as image recognition and object detection, 
the utility usually corresponds to inference accuracy, which typically exhibits concavity and other exploitable properties. 
However, in some scenarios, we may encounter more general utility functions, for which alternative strategies are required.

\subsubsection{RB Allocation with Concave Utility} \label{subsection:rb_allocation_concave}
If the utility function represents the inference accuracy or a concave function of the inference accuracy,
then it usually satisfies the following three properties.
Firstly, according to the results in 
\cite{yu2018slimmable, yu2019universally, yu2019autoslim, Cai2020Once-for-All, 
wang2022performance, zhang2023drl, zhang2023optimization, liu2023adaptable},
the marginal returns of both computational and communication workloads are diminishing.
Hence, we have:
\begin{assumption}
    $u^a(c,d)$ is non-decreasing and concave.
    \label{assumption1}
\end{assumption}
Secondly, the marginal return of computation workload should be non-decreasing with respect to the amount of transmitted data.
This is because any improvement in inference accuracy resulting from additional computation will be diminished if insufficient information 
is conveyed to the destination.
An extreme example is the marginal return $\frac{\partial u^a}{\partial c}$ remains zero when $d=0$, as $u^a(c,0) = 0$ for all $c$.
Therefore, we have:
\begin{assumption}
    $\frac{\partial u^a}{\partial c}$ is non-decreasing with respect to $d$.
    \label{assumption2}
\end{assumption}
Thirdly, according to the results in \cite{zhang2023drl}, the marginal return of $d$ drops rapidly,
which means the partial derivative of $u^a(c,d)$ with respect to $d$ also decreases rapidly.
Hence, we can assume $\frac{\partial u^a}{\partial d}$ satisfies the following property:
\begin{assumption}
    $\frac{\partial u^a}{\partial d}|_{(c,\alpha d)} \leq \frac{1}{\alpha} \frac{\partial u^a}{\partial d}|_{(c,d)}$ for any $\alpha \geq 1$.
    \label{assumption3}
\end{assumption}
This assumption requires $\frac{\partial u^a}{\partial d}$ decreases at least at a linear rate of $d$.
A common function that satisfies this property is the logarithmic function.
Although Assumption \ref{assumption3} may not hold strictly in all scenarios, 
it can be regarded as approximately valid within the range of values we are interested in.
A concrete utility model used in our numerical evaluation is provided in Section \ref{subsection:simulation_setup}, 
where we instantiate the task accuracy as a function of both computation and communication workloads
and discuss its consistency with the above assumptions.

Under these assumptions, we can proceed to demonstrate the concavity of $u^*_n(z^m_n)$.
Since the concavity for integer-valued variables is not well-defined,
we relax $z^m_n$ to a continuous variable $\hat{z}^m_n \in [0, K_m]$ and
prove that $u^*_n(\hat{z}^m_n)$ is concave.
%In the following text we will define a sequence of functions that gradually approach $u^*_n(z^m_n)$.
%We will prove that the first function in the sequence is concave.
%After that, we show that each function in the sequence inherits some degree of concavity from its predecessor. 
%As a result, the final function $u^*_n(z^m_n)$ can be regarded concave in a certain sense.
%For convenience, we first relax the integer constraint of $z^m_n$ and use $\hat{z}^m_n$ to denote corresponding real-value variable.
%Now we define our first function $\hat{u}^*_n(\hat{z}^m_n)$ as the optimal objective value of problem \eqref{subproblem}
%when the allocated RB is $\hat{z}^m_n$ and the utility is given by $u^a(c,d)$,
%i.e. $\hat{u}^*_n(\hat{z}^m_n) = \max_{(c,d)\in\Omega(\hat{z}^m_n)} u^a(c,d)$, where $\Omega(\hat{z}^m_n)$
%is the set of achievable computation and communication workload pairs when the allocated RB is $\hat{z}^m_n$.
%Now we prove that $\hat{u}^*_n(\hat{z}^m_n)$ is concave.
\begin{theorem}
    The function ${u}^*_n(\hat{z}^m_n)$ is concave.
    \label{theorem:concave_1}
\end{theorem}
\begin{IEEEproof}
    We only need to prove the second derivative is non-positive.
    For notational convenience, we temporarily drop the subscript $n$ and superscript $m$ in $\hat{z}^m_n$.
    Let $(c^*(\hat{z}), d^*(\hat{z}))$ be the optimal computation and communication workload pair when the allocated 
    number of RBs is $\hat{z}$.
    Since $u^a(c,d)$ is non-decreasing, we can safely assume $d^*(\hat{z}) = d^{max}(c^*(\hat{z}), \hat{z})$.
    Then we have
    \begin{align}
        \frac{\dif {u}^*_n(\hat{z})}{\dif \hat{z}} &= \frac{\dif u^a\left( c^*(\hat{z}), d^*(\hat{z}) \right)}{\dif \hat{z}} \notag \\
        &= \frac{\partial u^a}{\partial c^*} \times \frac{\dif c^*}{\dif \hat{z}} + \frac{\partial u^a}{\partial d^*} \times \frac{\dif d^*}{\dif \hat{z}} \notag \\
        &= \frac{\partial u^a}{\partial c^*} \times \frac{\dif c^*}{\dif \hat{z}} + \frac{\partial u^a}{\partial d^*} 
            \times \left( \frac{\partial d^*}{\partial c^*} \times \frac{\dif c^*}{\dif \hat{z}} + \frac{\partial d^*}{\partial \hat{z}} \right) \notag \\
        &= \frac{\dif c^*}{\dif \hat{z}} \left( \frac{\partial u^a}{\partial c^*} + \frac{\partial u^a}{\partial d^*} \times 
        \frac{\partial d^*}{\partial{c^*}}\right) + \frac{\partial u^a}{\partial d^*} \times \frac{\partial d^*}{\partial \hat{z}}. \label{eq:first_deriv}
    \end{align}
    According to \eqref{eq:d}, after plugging $f^*_n(c_n)$ into $d_n$,
    $d^{max}_n$ equals $z^m_n$ multiplies a function of $c_n$, denoted by $g(c_n)$.
    Therefore, we can express $d^*(\hat{z})$ as
    \begin{equation}
        d^*(\hat{z}) = \hat{z} g(c^*(\hat{z})).
        \label{eq:d_new}
    \end{equation}
    Since $c^*(\hat{z})$ is optimal, we have
    \begin{equation}
        \frac{\dif u^a}{\dif c^*} = \frac{\partial u^a}{\partial c^*} + \frac{\partial u^a}{\partial d^*} \times \frac{\partial d^*}{\partial c^*} = 0.
        \label{eq:c_zero_deriv}
    \end{equation}
    Substituting \eqref{eq:d_new} and \eqref{eq:c_zero_deriv} into \eqref{eq:first_deriv} yields 
    \begin{equation*}
      \frac{\dif {u}^*_n(\hat{z})}{\dif \hat{z}} = \frac{\partial u^a}{\partial d^*} \times g(c^*(\hat{z})).
    \end{equation*}
    Then the second derivative of ${u}^*_n(\hat{z})$ is
    \begin{equation}
      \frac{\dif^2 {u}^*_n(\hat{z})}{\dif \hat{z}^2} = \frac{\partial u^a}{\partial d^*} \times \frac{\dif g(c^*(\hat{z}))}{\dif \hat{z}}
        = \frac{\partial u^a}{\partial d^*} \times \frac{\partial g}{\partial c^*} \times \frac{\partial c^*}{\partial \hat{z}}. 
        \label{eq:second_deriv}
    \end{equation}
    Next, we show that $\frac{\dif^2 {u}^*_n(\hat{z})}{\dif \hat{z}^2}$ is non-positive by verifying the signs of the three terms on the 
    right-hand side of \eqref{eq:second_deriv}. 
    First, since $u^a(c,d)$ is non-decreasing, the first term $\frac{\partial u^a}{\partial d^*}$ is non-negative. 
    Next, because the communication and computation processes share the latency and energy budgets, 
    the maximum transmittable data $d^{\max}(c)$ is non-increasing with respect to the computation workload $c$. 
    According to \eqref{eq:d_new}, this further implies that $g(c^*(\hat{z}))$ is non-increasing with respect to $c^*(\hat{z})$, 
    and hence the second term $\frac{\partial g}{\partial c^*}$ is non-positive. 
    Therefore, it remains only to prove that the third term is non-negative.

    Let $\alpha > 1$ be an arbitrary constant.
    According to \eqref{eq:d_new}, it is easy to verify that $(c^*(\hat{z}), \alpha d^*(\hat{z}))$ is achievable when the allocated RB grows 
    from $\hat{z}$ to $\alpha \hat{z}$.
    The total derivative of $u^a(c,d)$ at $(c^*(\hat{z}), \alpha d^*(\hat{z}))$ satisfies
    \begin{align*}
        &\frac{\dif u^a}{\dif c^*}\Big|_{(c^*(\hat{z}), \alpha d^*(\hat{z}))} \\
        &\quad= \frac{\partial u^a}{\partial c^*}\Big|_{(c^*, \alpha d^*)} + \frac{\partial u^a}{\partial d^*}\Big|_{(c^*,\alpha d^*)} \times \frac{\partial d^*}{\partial c^*}\Big|_{(c^*, \alpha \hat{z})} \\
        &\quad\geq \frac{\partial u^a}{\partial c^*}\Big|_{(c^*, d^*)} + \frac{1}{\alpha} \frac{\partial u^a}{\partial d^*}\Big|_{(c^*,d^*)} \times \alpha \frac{\partial d^*}{\partial c^*}\Big|_{(c^*, \hat{z})} \\
        &\quad= \frac{\partial u^a}{\partial c^*}\Big|_{(c^*, d^*)} + \frac{\partial u^a}{\partial d^*}\Big|_{(c^*,d^*)} \times \frac{\partial d^*}{\partial c^*}\Big|_{(c^*, \hat{z})} = 0
    \end{align*}
    where the inequality is based on Assumption \ref{assumption2}, Assumption \ref{assumption3}, 
    $\frac{\partial d^*}{\partial c^*} \leq 0$, and the fact that $\frac{\partial d^*}{\partial c^*}$ is linear with respect to $\hat{z}$,
    and the last equality follows from \eqref{eq:c_zero_deriv}.
    Since the total derivative is non-negative, we will obtain a larger utility if we increase the value of $c^*(\hat{z})$.
    This indicates that, when the allocated RB is $\alpha \hat{z}$, the optimal computation workload is larger than $c^*(\hat{z})$,
    i.e. $c^*(\alpha \hat{z}) \geq c^*(\hat{z})$.
    Since the value of $\alpha$ is arbitrary, we can conclude that $c^*(\hat{z})$ is non-decreasing with respect to $\hat{z}$, 
    which completes the proof.
    % Please see Appendix \ref{appendix:3}.
\end{IEEEproof}

The concavity of the relaxed continuous function $u^*_n(\hat{z}^m_n)$ established in Theorem~\ref{theorem:concave_1} 
implies diminishing marginal returns for the original integer-valued utility function $u^*_n(z^m_n)$. 
Specifically, since $u^*_n(\hat{z}^m_n)$ is concave on $[0,K_m]$, 
its restriction to the integer points $z^m_n\in\{0,1,\dots,K_m\}$ satisfies the property that the marginal return of RB allocation
\begin{equation*}
\Delta_n(k)=u^*_n(k+1)-u^*_n(k)
\end{equation*}
is non-increasing with respect to $k$. 
Thus, problem \eqref{subproblem:allocation} becomes a separable discrete concave resource allocation problem, 
which can be optimally solved by greedily allocating RBs to the WD with the largest current marginal utility gain. 
The procedure is as follows.
Initially, we set $z^m_n = 0$ for all $n\in\mathcal{N}_m$.
At each iteration, we compute the marginal utility gain for allocating an additional RB to each WD,
i.e. $\Delta u_n = u^*_n(z^m_n+1) - u^*_n(z^m_n)$.
We then identify the WD $n' = \argmax_{n\in\mathcal{N}_m} \Delta u_n$
that yields the highest utility improvement and allocate one extra RB to it, i.e., $z^m_{n'} = z^m_{n'} + 1$.
This process is repeated until all $K_m$ RBs are allocated.

Our algorithm is summarized in Algorithm \ref{alg:rb_allocation_1}.
Note that in each iteration, only one variable $z^m_n$ is updated.
Therefore, except for the initial iteration, where all $\Delta u_n$ values must be computed, subsequent iterations only require updating 
$\Delta u_n$ for the previously selected WD. The overall time complexity of the algorithm is 
$O(|\mathcal{N}_m| K_m)$, where
$|\mathcal{N}_m|$ is the number of elements in $\mathcal{N}_m$.

\begin{algorithm}[t]
    \caption{Algorithm for the RB Allocation Subproblem with Concave Utility}
    \label{alg:rb_allocation_1}
    \begin{algorithmic}[1]
        %\renewcommand{\algorithmicrequire}{\textbf{Input:}}
        %\renewcommand{\algorithmicensure}{\textbf{Output:}}
        %\REQUIRE $\bm{x}_{-n}, p^a_k, s^a_k, \phi^a, \lambda^a_n$
        %\ENSURE  $t^{t,0}_i, t^{t,1}_i, t^p_i, t^{c,0}_i, t^{c,1}_i, x^{t}_{ij}, x^{c}_{ij}$
        \STATE Initialization: $z^m_n = 0, \forall n\in\mathcal{N}_m$;
        \FOR{$i = \{1, 2, \dots, K_m\}$}
            \FOR{$n\in\mathcal{N}_m$}
                \STATE Calculate $\Delta u_n = u^*_n(z^m_n+1) - u^*_n(z^m_n)$;
            \ENDFOR
            \STATE Find $n' = \argmax_{n\in\mathcal{N}_m} \Delta u_n$;
            \STATE Update $z^m_{n'} = z^m_{n'} + 1$;
        \ENDFOR
        \RETURN $z^m_n$.
    \end{algorithmic}
\end{algorithm}

\subsubsection{RB Allocation with General Utility}
For general utility functions, we cannot leverage the concavity of $u^*_n(\hat{z}^m_n)$ to facilitate algorithm design.
However, the RB allocation subproblem \eqref{subproblem:allocation} can still be effectively solved using dynamic programming.
Specifically, let $U_m(k,j)$ be the maximum aggregate utility attainable at SBS $m$ 
when allocating $k$ RBs among the first $j$ WDs, where $j \in \{ 1, \dots, |\mathcal{N}_m| \}$.
Formally, $U_m(k,j)$ corresponds to the optimal objective value of the following optimization problem:
\begin{align}
    \max_{z^m_n} \quad & \sum_{n=1}^j u^*_n(z^m_n) \label{prob:z} \\
    s.t.\quad & \sum_{n=1}^j z^m_n \leq k \notag \\
    & z^m_n \in \{0,1,\dots,k\}, \quad \forall n \in \{1,\dots,j\}. \notag
\end{align}

Apparently, solving the RB allocation subproblem \eqref{subproblem:allocation} is equivalent to computing $U_m(K_m,|\mathcal{N}_m|)$.
For notational convenience, we use $z^*_m(k,j)$ to denote the number of RBs allocated to the $j$-th WD under the optimal solution to problem \eqref{prob:z}.
Clearly, when there is only one WD in the system, we should allocate all RBs to it.
Therefore, we have
\begin{equation*}
    U_m(k,1) = u^*_1(k) \mbox{ and } z^*_m(k,1) = k.
\end{equation*}
For the case with two WDs, we consider all feasible RB allocations $z^m_2$ to the second WD and assign the remaining $k-z^m_2$ RBs to the first WD. 
The optimal allocation is thus given by:
\begin{gather*}
  U_m(k,2) = \max_{z^m_2} \left\{ U_m(k-z^m_2, 1) + u^*_j(z^m_2) \right\} \\
  z^*_m(k,2) = \argmax_{z^m_2} \left\{ U_m(k-z^m_2, 1) + u^*_j(z^m_2) \right\}.
\end{gather*}
This recursive formulation can be generalized to any $j\in\{2,3,\dots,N\}$.
Specifically, the dynamic programming recursion is given by:
\begin{gather}
  U_m(k,j) = \max_{z^m_j} \left\{ U_m(k-z^m_j, j-1) + u^*_j(z^m_j) \right\} \label{eq:z_m_j} \\
  z^*_m(k,j) = \argmax_{z^m_j} \left\{ U_m(k-z^m_j, j-1) + u^*_j(z^m_j) \right\}. \label{eq:z*}
\end{gather}
Based on \eqref{eq:z_m_j}, we can finally obtain $U_m(K_m, |\mathcal{N}_m|)$ in a bottom-up manner.
After that, the corresponding optimal RB allocation $z^m_n$ can be reconstructed via backtracking from the recorded $z^*_m(k,j)$.
Specifically, the backtracking procedure begins with $z^m_{|\mathcal{N}_m|} = z^*_m(K_m,|\mathcal{N}_m|)$ and 
for each $j\in\{1,\dots,|\mathcal{N}_m|-1\}$, we have 
\begin{equation*}
  z^m_j = z^*_m \left( K_m-\sum_{n=j+1}^{|\mathcal{N}_m|} z^m_n, j \right).
\end{equation*}

The detailed steps are summarized in Algorithm \ref{alg:rb_allocation_2}.
In each dynamic programming step \eqref{eq:z_m_j}, we iterate over all possible values of $z^m_j$, which amount to a total of $k$ values.
Therefore, the overall computational complexity of Algorithm \ref{alg:rb_allocation_2} is $O\left(|\mathcal{N}_m|K^2_m\right)$.

\begin{algorithm}[t]
    \caption{Algorithm for the RB Allocation Subproblem with General Utility}
    \label{alg:rb_allocation_2}
    \begin{algorithmic}[1]
        %\renewcommand{\algorithmicrequire}{\textbf{Input:}}
        %\renewcommand{\algorithmicensure}{\textbf{Output:}}
        %\REQUIRE $\bm{x}_{-n}, p^a_k, s^a_k, \phi^a, \lambda^a_n$
        %\ENSURE  $t^{t,0}_i, t^{t,1}_i, t^p_i, t^{c,0}_i, t^{c,1}_i, x^{t}_{ij}, x^{c}_{ij}$
        \STATE Initialization: $U_m(k,1)=u^*_1(k),\ z^*_m(k,1)=k,\ \forall k\in\{1,\dots,K_m\}$;
        \FOR{$j \in \{2,3,\dots,|\mathcal{N}_m|\}$}
            \FOR{$k \in \{0,1,\dots,K_m\}$}
              \STATE Calculate $U_m(k,j)$ and $z^*_m(k,j)$ according to \eqref{eq:z_m_j} and \eqref{eq:z*}.
            \ENDFOR
        \ENDFOR
        \STATE Let $k = K_m$;
        \FOR{$j = \{|\mathcal{N}_m|,|\mathcal{N}_m|-1,\dots,1\}$}
            \STATE Let $z^m_j = z^*_m(k,j)$;
            \STATE Let $k = k - z^m_j$;
        \ENDFOR
        \RETURN $z^m_n$.
    \end{algorithmic}
\end{algorithm}

\subsection{User Association Subproblem}
In this subsection, we address the optimization of the user association decision $\bm{x}$.
We begin by considering a relaxed setting where each WD is temporarily associated with all available SBSs, 
i.e. $\mathcal{N}_m = \mathcal{N}$ for all $m\in\mathcal{M}$.
Under this relaxed association, we apply the algorithms developed in previous subsections to compute the optimal RB allocation and scheduling decisions. 
The resulting total utility serves as an upper bound for the objective value of the original problem \eqref{problem}.

Subsequently, we iteratively refine the user association decision to satisfy the constraint \eqref{cons:association} through a sequence of $N$ steps. 
In each step, we select one WD that violates the association constraint and assign it to exactly one SBS. 
Our goal in each iteration is to minimize the resulting decrease in total utility. 
This process naturally raises two questions: (i) how to select the WD to process at each step, and (ii) how to determine its optimal SBS association.

We first address the second question.
Suppose we have selected WD $n'$ which is currently associated with all SBSs.
For an arbitrary SBS $m$, let $(z^m_{n})_{n\in\mathcal{N}_m}$ be its optimal RB allocation.
If WD $n'$ is disassociated from SBS $m$, the RBs it previously occupied can be redistributed to other WDs in $\mathcal{N}_m\setminus \{n'\}$.
The new optimal RB allocation is denoted by $(\bar{z}^m_{n})_{n\in\mathcal{N}_m\setminus \{n'\}}$.
Then the resulting decrease in utility for SBS $m$ is
\begin{align*}
\Delta_m &= \sum_{n\in\mathcal{N}_m} u^*_{n}(z^m_{n}) - \sum_{n\in\mathcal{N}_m\setminus \{n'\}} u^*_{n}(\bar{z}^m_{n})  \\
         &= u^*_{n'}(z^m_{n'}) - \sum_{n\in\mathcal{N}_m\setminus \{n'\}} \left( u^*_{n}(\bar{z}^m_{n}) - u^*_{n}(z^m_{n}) \right).
\end{align*}
Hence, if we eventually associate WD $n'$ with SBS $m'$ (which means WD $n'$ will disassociate with the rest SBSs), 
then the decrease in total utility will be $\Delta^t = \sum_{m\in\mathcal{M}\setminus \{m'\}} \Delta_m$.
Recall that our objective is to minimize $\Delta^t$, thus the optimal $m' = \argmin_{m'} \Delta^t = \argmax_{m'} \Delta_{m'}$.

Since the differences between $z^m_{n}$ and $\bar{z}^m_{n}$ are typically small, we can employ linear approximation to estimate
the change in utility due to a variation in the RB allocation:
\begin{equation*}
u^*_{n}(\bar{z}^m_{n}) - u^*_{n}(z^m_{n}) \approx \frac{\dif u^*_{n}}{\dif z^m_{n}} \times (\bar{z}^m_{n} - z^m_{n}).
\end{equation*}
Since $z^m_n$ is the optimal RB allocation, the Karush-Kuhn-Tucker (KKT) conditions imply that
the derivative $\frac{\dif u^*_{n}}{\dif z^m_{n}}$ should be equal for all $n\in\mathcal{N}_m$. \footnote{
  Strictly speaking, the KKT conditions apply only to continuous variables.
  While $z^m_n$ takes discrete values, the corresponding derivatives are typically approximately equal in practice.
}
Therefore, it suffices to evaluate the derivative for a single WD.
Since calculating derivatives for discrete variables is often inconvenient or ill-defined, 
we approximate the derivative using the marginal gain:
\begin{equation*}
  \delta_m = u^*_{n}(z^m_{n}+1) - u^*_{n}(z^m_{n}).
\end{equation*}
Consequently, the utility reduction $\Delta_m$ can be expressed as:
\begin{align}
  \Delta_m &\approx u^*_{n'}(z^m_{n'}) - \sum_{n\in\mathcal{N}_m\setminus \{n'\}} \delta_m (\bar{z}^m_{n} - z^m_{n}) \notag \\
           &= u^*_{n'}(z^m_{n'}) - \delta_m z^m_{n'} \label{eq:Delta_m}
\end{align}
where the equality holds because $\sum_{n\in\mathcal{N}_m\setminus \{n'\}} \bar{z}^m_{n} = K_m$
and $\sum_{n\in\mathcal{N}_m\setminus \{n'\}} {z}^m_{n} = K_m - z^m_{n'}$.

The next problem is how to select the appropriate WD to process at each step.
As described above, for any given WD $n$, we choose its target SBS based on $\Delta_m$.
However, the value of $\Delta_m$ depends on the set of WDs associated with SBS $m$, which is $\mathcal{N}_m$.
As WDs progressively assigned to specific SBSs, the set $\mathcal{N}_m$ evolves,
potentially altering the values of $\Delta_m$.
Consequently, a previously optimal SBS for WD $n$ may no longer be optimal later---a phenomenon we term \emph{regret}.

To mitigate regret, we make the following two observations.
First, if the value of $\Delta_m$ for one SBS significantly exceeds that of others, 
it is likely that the SBS remains optimal even as $\mathcal{N}_m$ changes. 
Second, if WD $n$ is processed late in the sequence, then the subsequent changes in $\mathcal{N}_m$ are limited.
This implies that the variation of $\Delta_m$ is also limited, thereby reducing the risk of regret.

Based on these observations, we prioritize WDs for which the values of $\Delta_m$ differ most distinctly.
To quantify this, we define the kurtosis of WD $n$ as:
\begin{equation*}
\kappa_n = \frac{\max_m \Delta_m}{\sum_m \Delta_m}.
\end{equation*}
A lower $\kappa_n$ indicates that $\Delta_m$ values are more uniform across SBSs, 
implying higher uncertainty in SBS selection. 
Hence, in each step, we select the WD with the largest $\kappa_n$, 
thereby deferring those with higher risk of regret.

The overall procedure is outlined in Algorithm \ref{alg:association},
where $\mathcal{N}_r$ denotes the set of WDs to be processed.
Let $K = \max_m K_m$ be the maximum number of RBs in all SBSs.
Based on the result in the previous subsection, calculating the optimal RB allocation
in line $3$ has a time complexity of $O(MNK)$ for concave utility functions
and $O(MNK^2)$ for general ones.
The time complexity of user selection and update process in line $4$-$9$ is $O(NM)$.
Thus, the total time complexity of Algorithm \ref{alg:association} is
$O(MN^2K)$ for concave utility functions and $O(MN^2K^2)$ for general utility functions.

\emph{Remark:} We would like to clarify that the proposed three-stage framework is not an alternating optimization procedure 
that repeatedly optimizes the three subproblems until convergence.
Instead, the per-WD scheduling and RB allocation stages are invoked as deterministic subroutines under a given association decision, 
while the user association stage performs a finite relaxation-and-refinement process.
In Algorithm \ref{alg:association}, each outer iteration fixes the association of exactly one WD and removes it from $\mathcal{N}_r$.
Therefore, $|\mathcal{N}_r|$ decreases monotonically and the algorithm terminates after exactly $N$ iterations.
% The resulting association is feasible and stable in the sense that all WDs satisfy \eqref{cons:association}, 
% after which the RB allocation and scheduling decisions are computed consistently with this final association.
This finite-termination property does not imply global optimality, as the refinement step relies on marginal-utility approximation, 
but it guarantees that the proposed procedure produces a well-defined solution in polynomial time.

\begin{algorithm}[t]
    \caption{Algorithm for the User Association Subproblem}
    \label{alg:association}
    \begin{algorithmic}[1]
        %\renewcommand{\algorithmicrequire}{\textbf{Input:}}
        %\renewcommand{\algorithmicensure}{\textbf{Output:}}
        %\REQUIRE $\bm{x}_{-n}, p^a_k, s^a_k, \phi^a, \lambda^a_n$
        %\ENSURE  $t^{t,0}_i, t^{t,1}_i, t^p_i, t^{c,0}_i, t^{c,1}_i, x^{t}_{ij}, x^{c}_{ij}$
        \STATE Initialization: $\mathcal{N}_r = \mathcal{N}, x^m_n = 1, \forall n\in\mathcal{N}, m\in\mathcal{M}$;
        \FOR{$i\in\{1,\dots,N\}$}
            \STATE Solve the RB allocation subproblem for each $m\in\mathcal{M}$;
            \FOR{$n \in \mathcal{N}_r$}
                \FOR{$m\in\mathcal{M}$}
                    \STATE Estimate $\Delta_m$ according to \eqref{eq:Delta_m};
                \ENDFOR
                \STATE Calculate $\kappa_n = \frac{\max_m \Delta_m}{\sum_m \Delta_m}$;
            \ENDFOR
            \STATE Find $n' = \argmax_n \kappa_n$ and $m' = \argmax_m \Delta_m$;
            \STATE Set $x^{m'}_{n'} = 1$ and $x^m_{n'} = 0$ for all $m\ne m'$;
            \STATE Update $\mathcal{N}_r = \mathcal{N}_r \setminus \{n'\}$;
        \ENDFOR
        \RETURN $x^m_n$.
    \end{algorithmic}
\end{algorithm}

\section{Numerical Results} \label{section:simulation}
\begin{table}[t]
% increase table row spacing, adjust to taste
\renewcommand{\arraystretch}{1.1}
 %if using array.sty, it might be a good idea to tweak the value of
 %\extrarowheight as needed to properly center the text within the cells
\caption{Simulation Parameters}
\label{table:sim_param}
\centering
% Some packages, such as MDW tools, offer better commands for making tables
% than the plain LaTeX2e tabular which is used here.
\begin{tabularx}{0.9\linewidth}{l|c}
\hline
\textbf{Parameter} & \textbf{Value}\\
\hline
Number of RBs in each SBS $K_m$ & $\mathcal{U}\{10, 20, 30, 40\}$ \\
Bandwidth of each RB $W$ & $0.2$ MHz \\
Delay requirement $T^{max}$ & $10$ ms \\
Energy budget of WDs $E^{max}_n$ & $2$ mJ \\
Maximum CPU frequency $f^{max}_n$ & $\mathcal{U}[1,3]$ GHz \\
Maximum transmission power $P^{max}_n$ & $0.2$ W \\
Energy efficiency of WDs' processors $\gamma_n$ & $\mathcal{U}[10^{-28}, 10^{-27}]$ \\
Noise power at SBSs $\sigma^2$ & $10^{-13}$ \\
Average interference at SBSs $I_m$ & $\mathcal{U}[10^{-13}, 10^{-12}]$ \\
\hline
\end{tabularx}
\end{table}

In this section, extensive simulations are conducted to evaluate the performance of our algorithm
and numerical results show that our method significantly outperforms the conventional communication paradigm and existing semantic compression approach.

\subsection{Simulation Setup} \label{subsection:simulation_setup}
We consider a 5G network consisting of $M=10$ SBSs and $N=100$ WDs
that are randomly located in a $1000$m$\times 1000$m square area.
For the channel model, we assume the pathloss is $128.1 + 37.6 \log_{10}Dist$ dB
and the shadowing factor is $6$ dB, where $Dist$ is the distance (in km) between WDs and SBSs.

Our utility function is based on the empirical results about the relationship
between result accuracy and computational or communication workload.
To model the impact of computational workload $c$ on result accuracy, we conduct curve-fitting on the results 
given in \cite{yu2019universally}, which shows the accuracy of various models under different computational overhead.
We find that the accuracy $A_c$ can be well approximated by the following logarithmic function
\begin{equation*}
  A_c(c) = \eta^a_1 \log\left(\frac{c}{C^a}\right) + \eta^a_2,
\end{equation*}
where $\eta^a_1, \eta^a_2$ are parameters corresponding to application $a$ 
and $C^a$ is the computational workload of the maximum DNN model.
In our simulation, we set $\eta^a_1\in\mathcal{U}[0.05, 0.08], \eta^a_2\in\mathcal{U}[0.9, 0.95]$, and $C^a \in\mathcal{U}[5M, 10M]$ cycles,
where $\mathcal{U}[a,b]$ represents the uniform distribution on interval $[a,b]$.

According to \cite{zhang2023drl}, the relationship between result accuracy $A_d$ and communication workload $d$ can be modeled by the following function:
\begin{equation*}
  A_d(d) = \beta^a_1 (1- \frac{d}{D^a})^{\beta^a_2} + \beta^a_3,
\end{equation*}
where $\beta^a_1, \beta^a_2, \beta^a_3$ are parameters and $D^a$ is the maximum data size of the DNN outputs.
In our simulation, we set $\beta^a_1 \in \mathcal{U}[-0.75, -0.6], \beta^a_2 \in \mathcal{U}[10, 20],
\beta^a_3 \in \mathcal{U}[0.9, 0.95]$ and $D^a \in \mathcal{U}[0.15M, 0.25M]$ bits.
Intuitively, $\beta^a_3$ is the result accuracy if all DNN outputs are transmitted to the receiver.
Hence, the accuracy discount due to limited data transmission is $A_d(d)/\beta^a_3$.
Therefore, we assume the final result accuracy at the receiver side is 
\begin{equation*}
  A(c,d) = A_c(c) \times \frac{A_d(d)}{\beta^a_3}.
\end{equation*}

To evaluate our algorithm comprehensively, we consider a concave utility function $u^a(c,d) = A(c,d)$
and a general non-concave function $u^a(c,d) = \frac{1}{1-A(c,d)}$.
The concave utility function is consistent with the assumptions in Section \ref{subsection:rb_allocation_concave}.
Specifically, with $\eta^a_1>0$, $\beta^a_1<0$, and $\beta^a_2>1$, 
both $A_c(c)$ and $A_d(d)$ are non-decreasing and exhibit diminishing marginal returns, which is consistent with Assumption \ref{assumption1}.
Moreover, $\partial u^a/\partial c = \eta^a_1 A_d(d)/(\beta^a_3 c)$ is non-decreasing with respect to $d$, which is consistent with Assumption \ref{assumption2}, 
and the adopted $A_d(d)$ captures the rapidly decreasing marginal gain of transmitted semantic components, which is consistent with Assumption \ref{assumption3}.
% In our numerical implementation, the induced discrete utility sequence $u^*_n(z)$ generated under this accuracy-based utility 
% is verified to have non-increasing marginal gains over the considered RB range, so the greedy RB allocation is applied to this setting.
% The second utility, $u^a(c,d)=1/(1-A)$, is intentionally non-concave and is evaluated using the dynamic programming-based RB allocation algorithm for general utility functions.
The values of rest parameters are listed in Table \ref{table:sim_param}.

To demonstrate the effectiveness of the proposed algorithm (labeled as ``Prop''), we compare it with the following benchmarks.
Notice that our algorithm consists of three stages.
Unless otherwise specified, each benchmark only differs from our algorithm in one stage.
\begin{itemize}
    \item Traditional Communication (TC): the raw data is transmitted to the SBS without processing.
    \item Fixed Semantic Communication (FSC): each WD adopts a fixed SemCom scheme.
    \item Semantic Compression (SC): fixed computation with selective transmission of semantic information.
    \item Average RB Allocation (ARB): the RBs are allocated to each associated WD evenly.
    \item Nearest User Association (NUA): the WDs are associated with the nearest SBSs.
\end{itemize}
For TC, we assume the size of raw data satisfies $D_{raw} \in \mathcal{U}[0.4M, 0.8M]$ bits.
If WD $n$ successfully transmits the raw data under the delay and energy constraints, then 
the SBS can use the maximum DNN model to process the raw data, which leads to result accuracy $\eta^a_2$.
For FSC, the utilized fixed SemCom scheme is $c = C^a/2$ and $d = D^a/2$.
In SC schemes \cite{wang2022performance, zhang2023drl, liu2023adaptable}, the complete semantic representation is first extracted, 
resulting in $c = C^a$. 
Subsequently, only the most critical semantic information is transmitted, 
subject to constraints on communication resources, latency, and energy consumption. 
To further validate the effectiveness of our design, we introduce two ablation baselines, 
namely ARB and NUA, which are employed to isolate and highlight the contributions of the 
proposed RB allocation and user association strategies, respectively.
The data presented in this section are the average of $30$ repeated experiments with different random seeds.

\begin{figure}[t]
\centering
\includegraphics[width=3.2in]{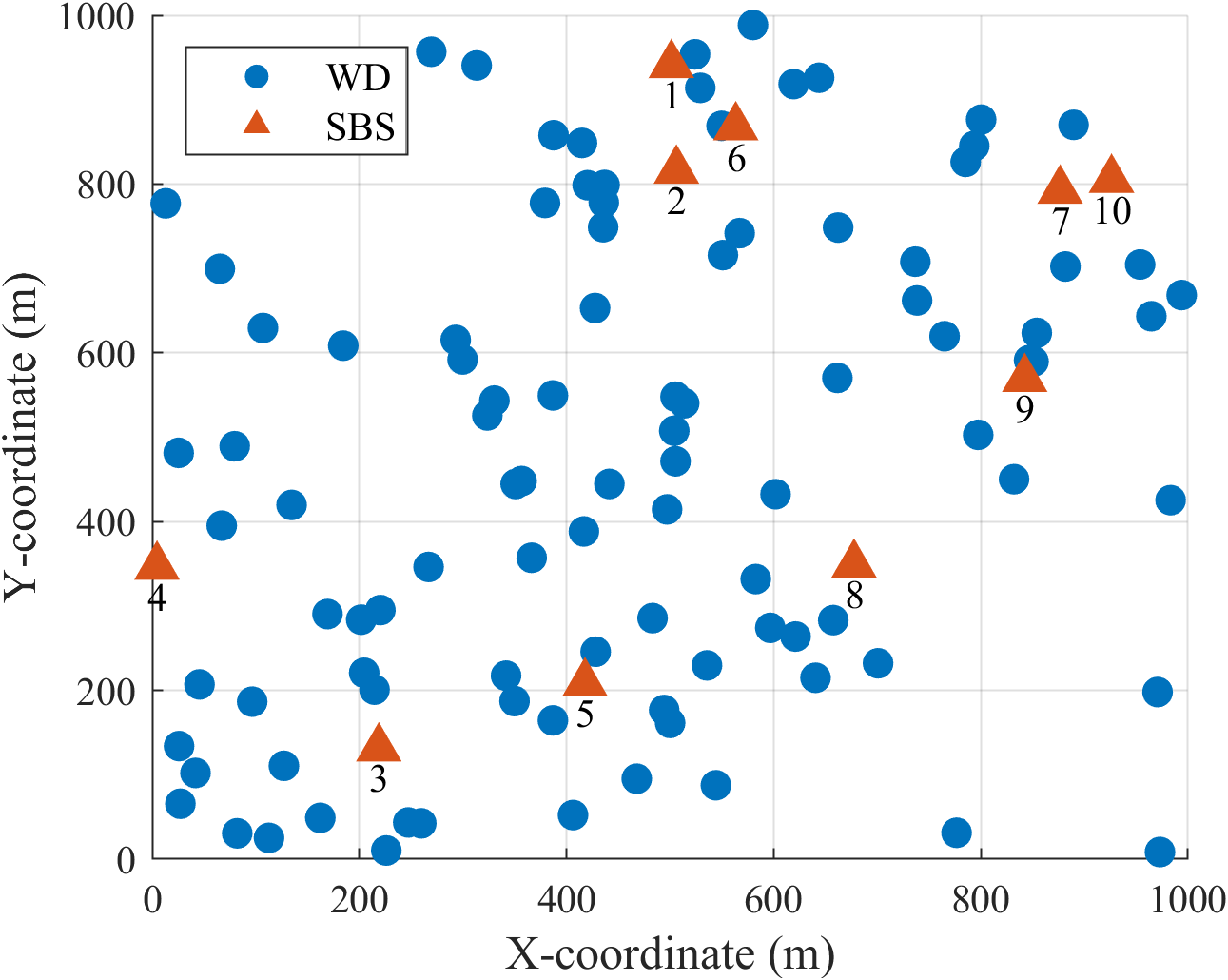}
\caption{Locations of WDs and SBSs.}
\label{fig:location}
\end{figure}

\begin{figure}[t]
\centering
\includegraphics[width=3.2in]{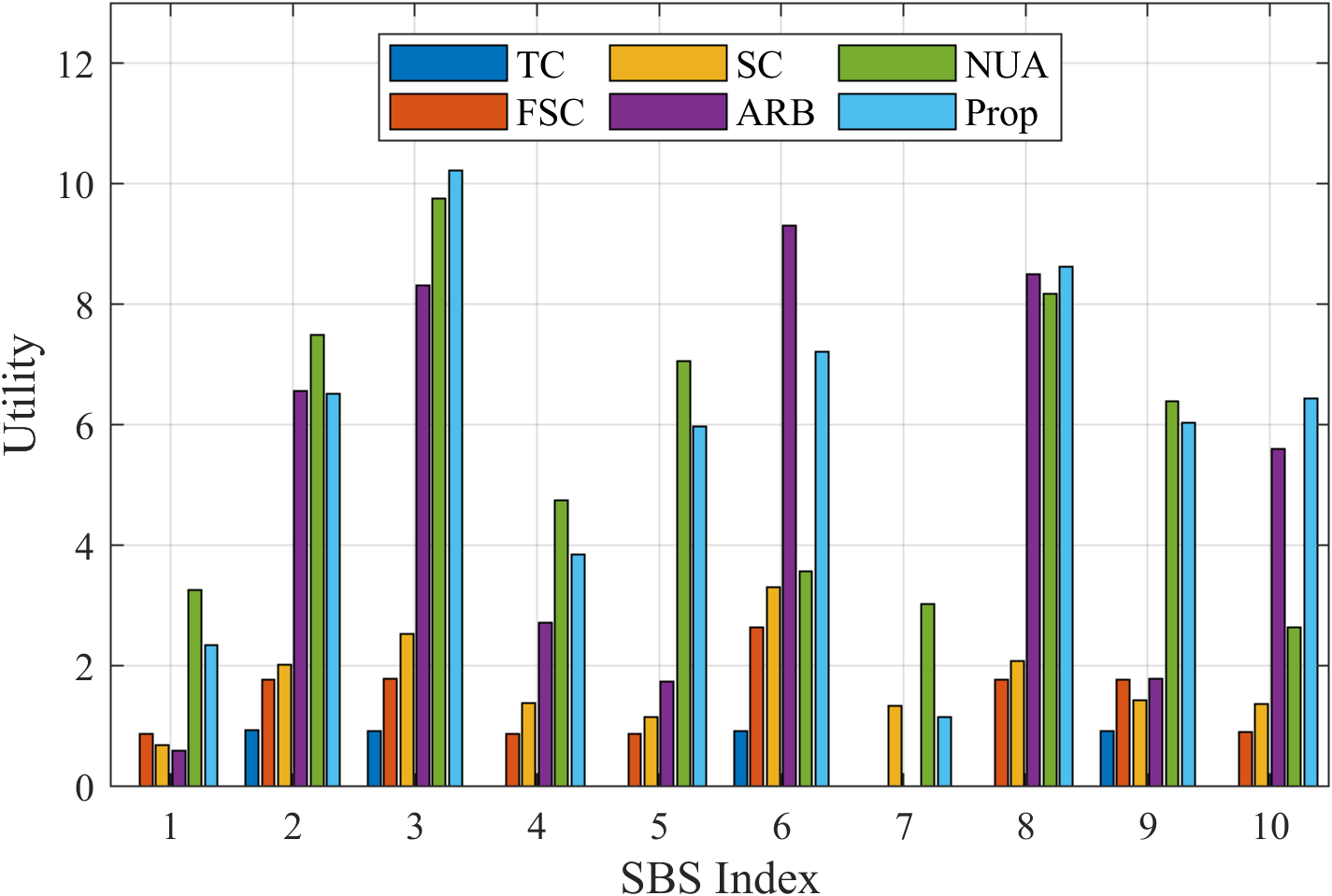}
\caption{Utility of each SBS under the concave utility function.}
\label{fig:concave}
\end{figure}

% \begin{figure}[t]
% \includegraphics[width=3.4in]{fig/general.png}
% \caption{Utility of each SBS under the general utility function.}
% \label{fig:general}
% \end{figure}

\subsection{Performance Comparison under a Specific Scenario}
To elucidate the strategic differences among the considered algorithms, 
we first analyze system performance under a representative scenario.
Fig. \ref{fig:location} illustrates the randomly generated positions of WDs 
and SBSs, while Fig. \ref{fig:concave} presents the utility achieved by each SBS under the concave utility function. 
In our simulations, the available wireless resources are highly constrained relative to the volume of raw data. 
As a result, either none or only a limited fraction of WDs can successfully transmit their raw data to SBSs, 
leading to degraded system utility under TC. 
The adoption of SemCom alleviates this issue by substantially reducing the amount of data to be transmitted. 
Within this framework, SC outperforms FSC owing to its more adaptable transmission strategy. 
Nevertheless, SC still suffers from fixed computational requirements, resulting in excessive latency and energy consumption.

In contrast, ASC achieves a notable performance gain. 
Due to the heterogeneity of WDs, equal allocation of RBs across users is inefficient. 
Consequently, although the ablation baseline ARB achieves higher utility for certain SBSs, 
the overall network utility remains inferior to that of the proposed approach. 
A similar observation applies to the NUA strategy, which tends to yield near-optimal performance in 
regions with sparsely distributed SBSs. 
However, in areas where SBSs are densely located, systematic user association becomes crucial. 
For instance, by associating some WDs in proximity to SBSs $1$ and $2$ with SBS $6$, 
the proposed method further enhances overall network utility.

\begin{figure}[t]
\centering
\includegraphics[width=3.4in]{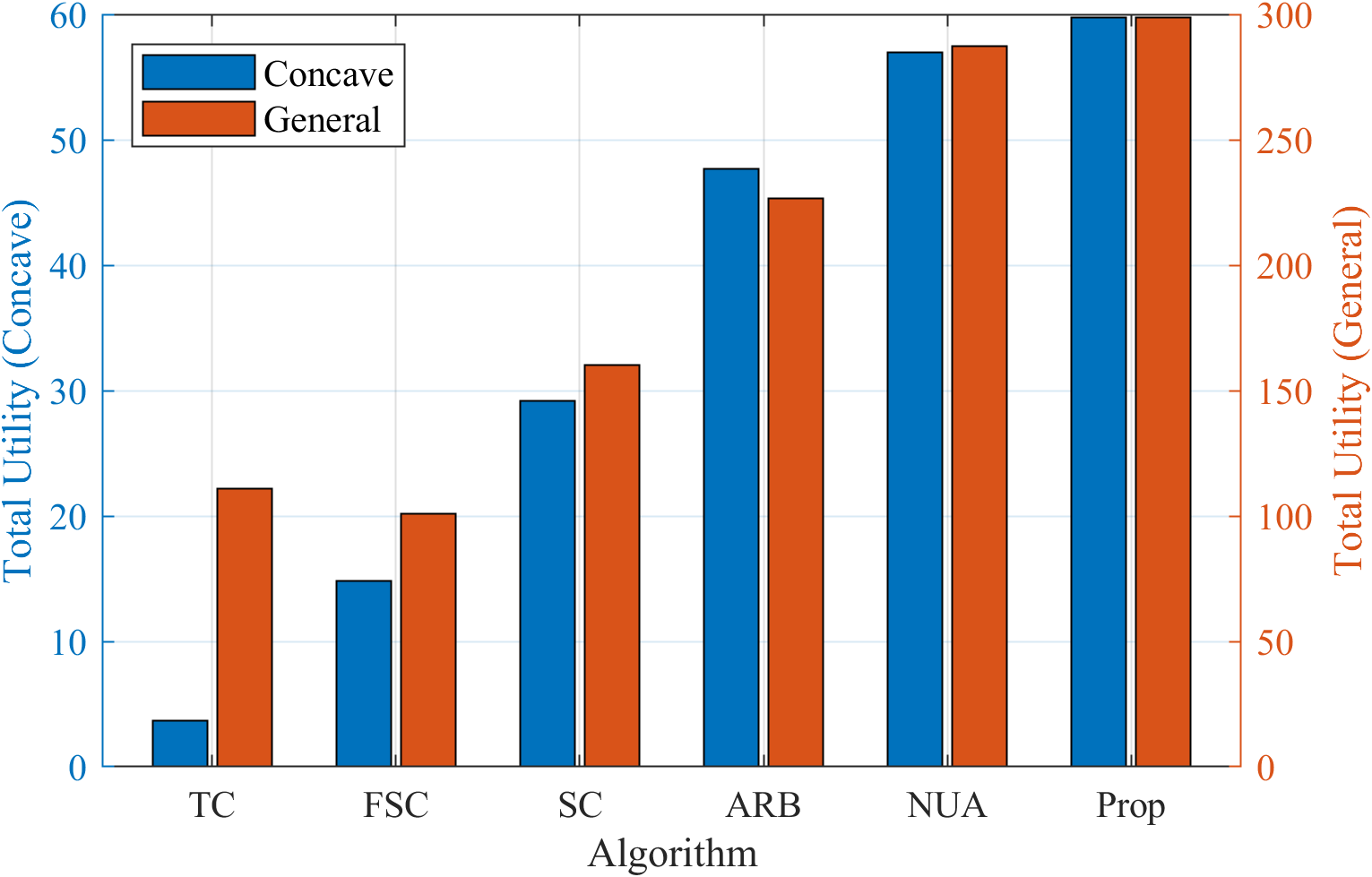}
\caption{Total utility of different algorithms.}
\label{fig:total_u}
\end{figure}

\subsection{Performance under Different Utility Functions}
To provide a comprehensive comparison of the overall performance of the considered algorithms, 
Fig. \ref{fig:total_u} depicts their performance under both concave and general utility functions. 
For the concave utility function, the performance trends of the algorithms are consistent with the case study presented in the previous subsection. 
Moreover, most algorithms exhibit comparable performance across the two utility functions, with the exception of TC. 
This divergence arises because, under the general utility function, utility is defined as the reciprocal of one minus accuracy, 
which increases sharply as accuracy approaches unity. 
Consequently, TC benefits from this property, as it achieves particularly high utility when a WD successfully transmits the raw data to the SBS. 
In both utility settings, the proposed algorithm consistently achieves the highest utility, 
thereby demonstrating its robustness and near-optimality across different classes of utility functions.

\begin{figure}[t]
\centering
\includegraphics[width=3.2in]{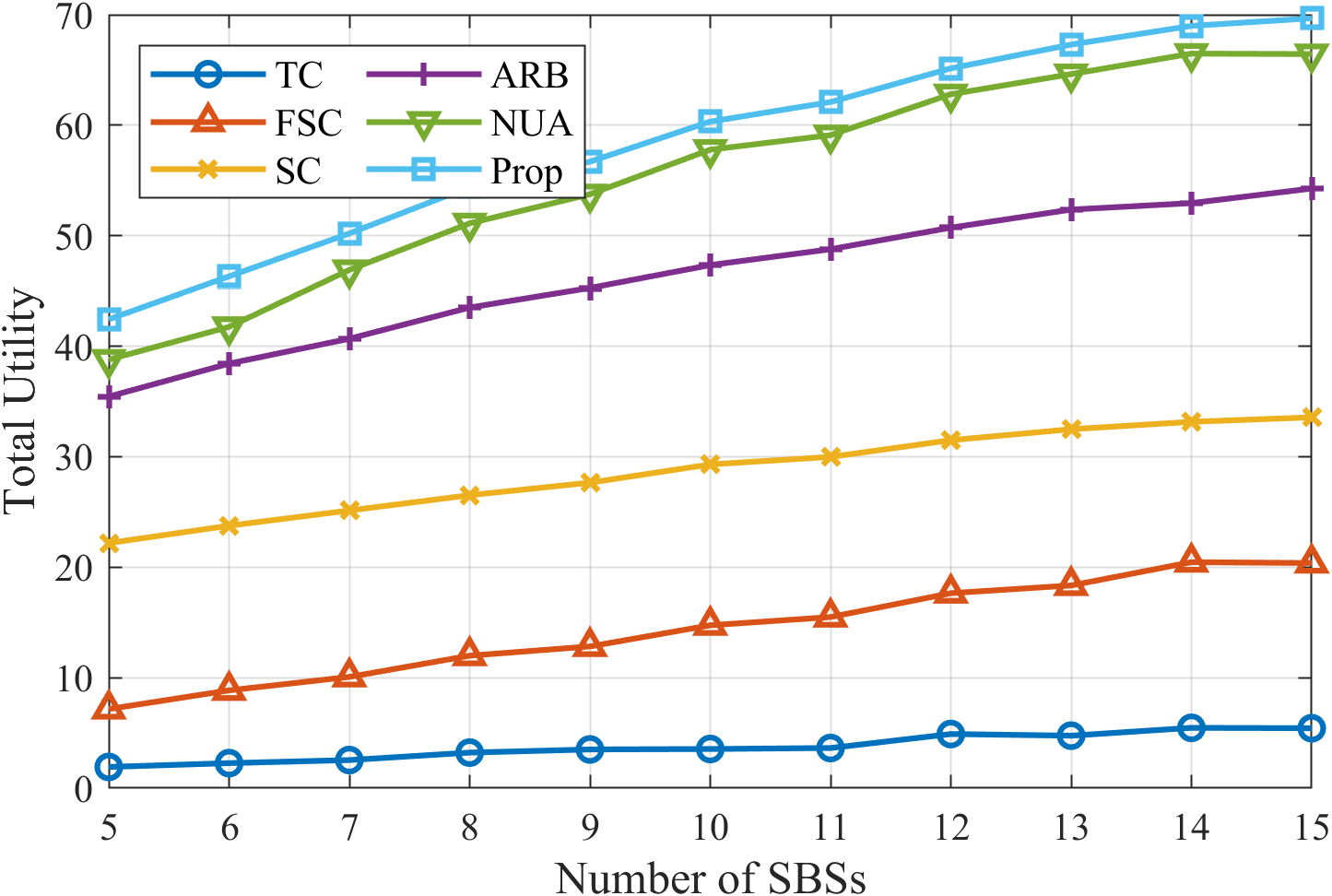}
\caption{Total utility under different numbers of SBSs.}
\label{fig:bs}
\end{figure}

\begin{figure}[t]
\centering
\includegraphics[width=3.2in]{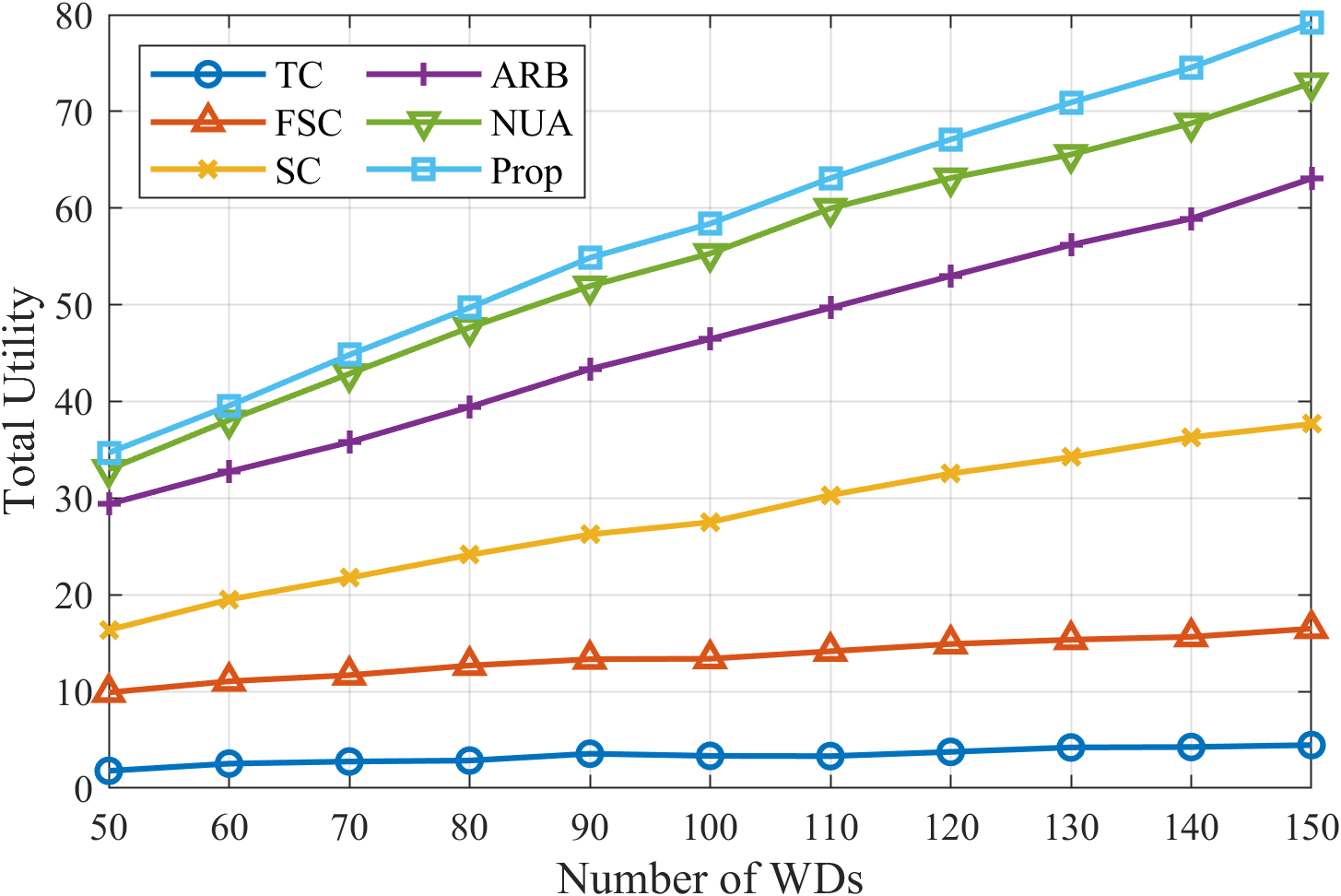}
\caption{Total utility under different numbers of WDs.}
\label{fig:wd}
\end{figure}

\subsection{Impact of Network Scale}
To evaluate the influence of network scale on algorithmic performance, Fig. \ref{fig:bs} and Fig. \ref{fig:wd} present the total utility of 
all algorithms under varying numbers of SBSs and WDs. 
For clarity, only results under the concave utility function are shown, as similar trends are observed for the general utility function. 
As illustrated in Fig. \ref{fig:bs}, increasing the number of SBSs reduces the average distance between WDs and SBSs, 
thereby improving channel gains. 
Moreover, a larger number of SBSs also provides additional RBs. 
Consequently, the utility achieved by all algorithms increases with the number of SBSs. 
However, the performance improvement of TC remains marginal, highlighting its inefficiency. 
Among the ASC-based schemes, ARB exhibits a smaller performance gain compared with the proposed algorithm and NUA, 
primarily because its average RB allocation strategy fails to fully exploit the additional RBs, thereby diminishing the benefits of network densification.

A similar trend is observed when varying the number of WDs. 
Given the limited wireless resources, systematic resource allocation becomes critical to achieving higher utility. 
Nonetheless, since both TC and FSC employ fixed transmission strategies, they lack enough flexibility in resource allocation, 
resulting in minimal performance gains. 
For the same reason, although SC adopts an adaptive transmission strategy, its fixed computation mechanism restricts flexibility compared to ASC paradigms. 
Consequently, its performance improvement is also limited relative to the proposed algorithm, NUA, and ARB.

\begin{figure}[t]
\centering
\includegraphics[width=3.2in]{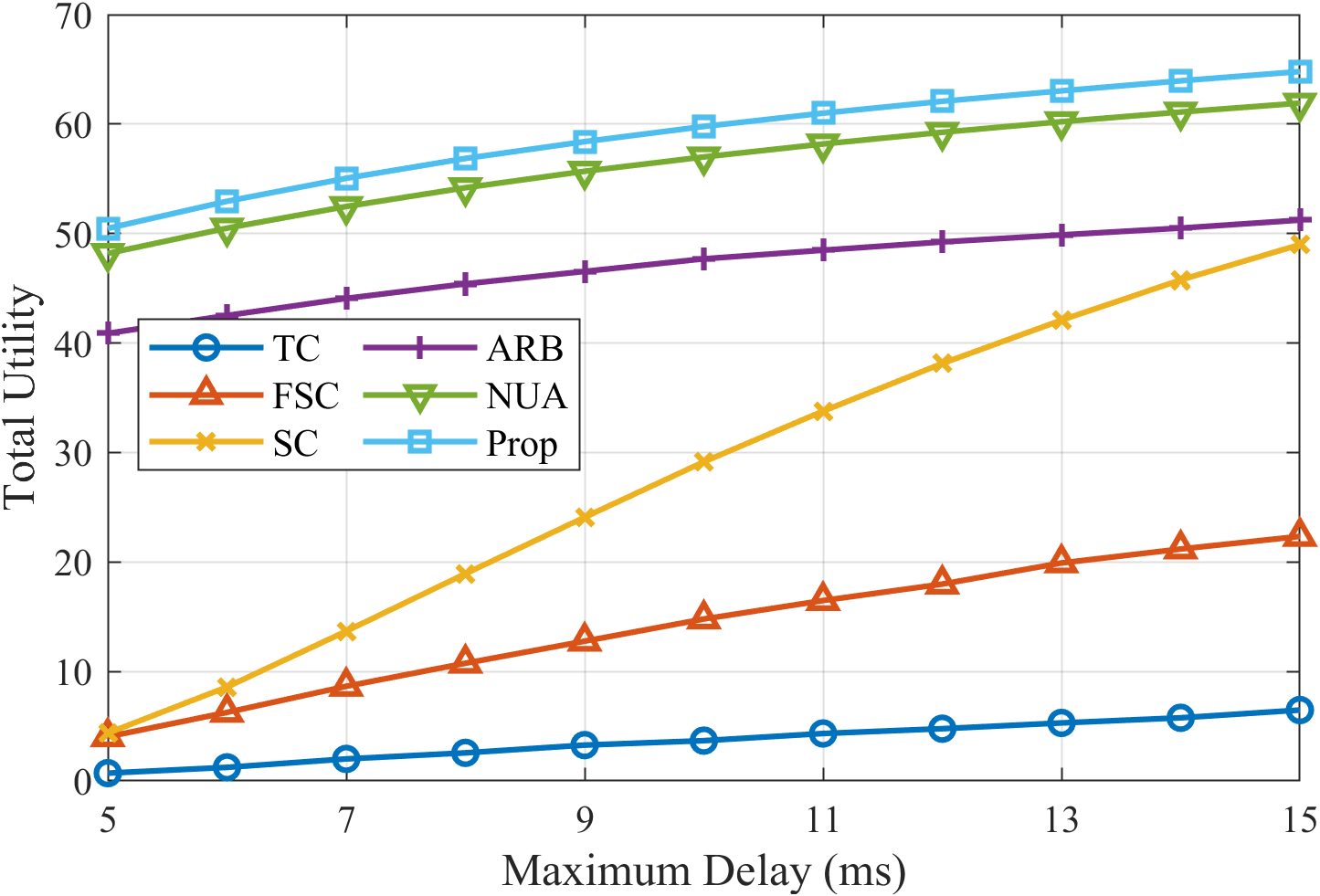}
\caption{Total utility under different delay requirements.}
\label{fig:delay}
\end{figure}

\begin{figure}[t]
\centering
\includegraphics[width=3.2in]{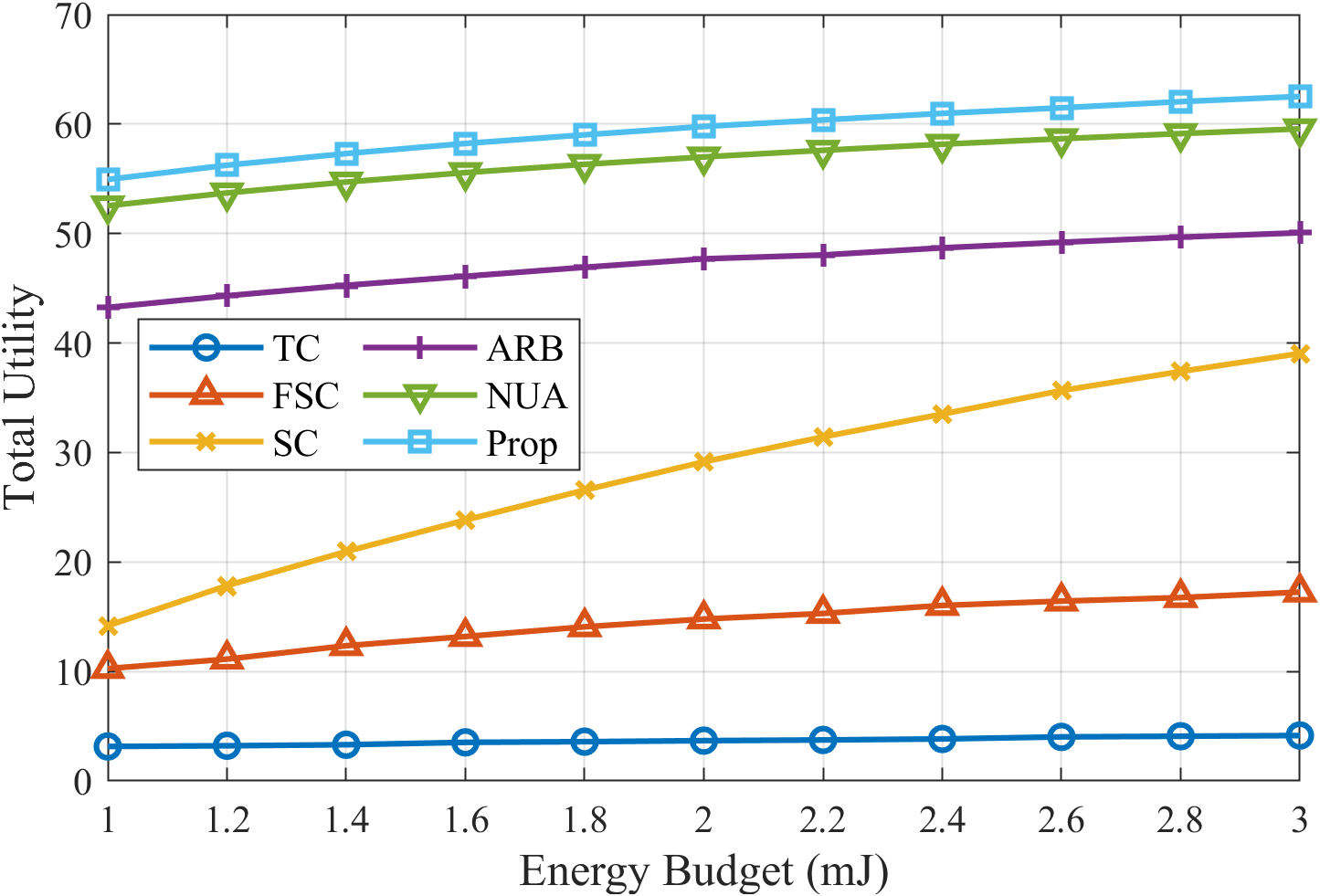}
\caption{Total utility under different energy budgets.}
\label{fig:energy}
\end{figure}

\subsection{Impact of Delay and Energy Constraints}
In addition to network scale, we further examine the effects of delay and energy constraints on algorithmic performance, 
as depicted in Fig. \ref{fig:delay} and Fig. \ref{fig:energy}. 
Among all schemes, SC exhibits the highest sensitivity to both constraints. 
This behavior arises because SC employs a fixed computation process that consumes a constant amount of time and energy; 
thus, any additional delay or energy budget can be entirely utilized for transmitting more semantic information, resulting in substantial marginal gains. 
Furthermore, the total utility of all algorithms is observed to be more sensitive to delay requirements than to energy budgets. 
This outcome is not surprising, as the transmitted data volume scales linearly with transmission time but only logarithmically with transmission power.
Notice that our algorithm constantly outperforms all benchmarks in all cases, demonstrating its effectiveness under diverse scenarios.

\subsection{Practical Running Time}
\begin{table}[t]
\centering
\caption{Running Time of Algorithms}
\label{tab:runtime}
\setlength{\tabcolsep}{4pt}
\begin{tabular}{ccccccc}
\toprule
\multirow{2}{*}{Algorithm} & \multicolumn{6}{c}{Number of WDs ($N$)} \\
\cmidrule(lr){2-7}
& 50 & 60 & 70 & 80 & 90 & 100 \\
\midrule
Prop & 0.431 & 0.466 & 0.583 & 0.749 & 0.826 & 0.993 \\
FSC & 7.123 & 10.032 & 13.774 & 18.212 & 22.700 & 28.295 \\
ARB & 0.246 & 0.375 & 0.409 & 0.525 & 0.598 & 0.755 \\
NUA & 0.004 & 0.004 & 0.005 & 0.005 & 0.005 & 0.005 \\
SC & 4.849 & 6.681 & 9.195 & 12.024 & 15.029 & 18.043 \\
TC & 0.156 & 0.207 & 0.296 & 0.408 & 0.476 & 0.585 \\
\bottomrule
\end{tabular}
\end{table}

Table \ref{tab:runtime} reports the average running time, in seconds, under the concave utility function.
All algorithms are implemented in MATLAB and executed on a Windows 10 PC with an Intel Core i5-11600 CPU and 16 GB RAM.
The proposed algorithm solves the problem with $100$ WDs in about one second on average in this unoptimized MATLAB implementation.
This result provides a baseline measurement of the computational cost and 
the running time can be significantly improved in realistic deployment, as discussed in the following.

The measured runtime trends are generally consistent with the complexity analysis:
for fixed $M$ and $K$, the dominant association procedure scales approximately quadratically with the number of WDs.
The proposed algorithm grows slightly slower than a strict quadratic trend in this range, 
mainly because fixed function-call overhead from repeated RB allocation and per-WD scheduling subroutines is non-negligible in MATLAB.
FSC and SC require longer runtime because they do not use the proposed per-WD ASC scheduling step.
As a result, their effective utility with respect to RB allocation is not guaranteed to be concave even when the original utility function is concave, so the dynamic-programming-based RB allocation algorithm is used instead of the greedy one.
In contrast, ARB and NUA are faster because they replace optimized RB allocation or user association with heuristic rules, 
but the previous results show that this speedup comes at the cost of lower system utility.

In practical deployment, the running time can be further reduced through implementation-level and algorithm-level acceleration.
The current implementation is written in MATLAB mainly for algorithm validation; 
a compiled implementation in C/C++, vectorized computation, and parallel evaluation of per-WD scheduling and 
RB allocation across SBSs can reduce the runtime substantially.
Moreover, when stricter real-time requirements are imposed, 
coarser RB grouping or lightweight association heuristics can be adopted to trade a 
small amount of optimality for lower decision latency.

\section{Conclusion} \label{section:conclusion}
In this paper, we have studied the joint optimization of user association and resource allocation for ASC in 5G and beyond networks.
We discussed the feasibility of ASC and formulated the utility maximization problem under limited energy and delay.
An efficient three-stage algorithm has been proposed by leveraging the structural characteristics of each subproblem and their interconnections.
Numerical results indicated that the performance of SemCom can significantly outperform the 
traditional communication paradigm if we choose appropriate control decisions.
% In our future work, we will embark on building a real ASC system and 
% refine our algorithm based on the experimental results in practical scenarios.

% \appendices
%
% \renewcommand{\thesectiondis}[2]{\Alph{section}:}
%
% \section{Proof of Theorem \ref{theorem:np}} \label{appendix:1}
%
% \section{Proof of Theorem \ref{prop:tightness}} \label{appendix:2}
%
% \section{Proof of Theorem \ref{theorem:concave_1}} \label{appendix:3}

\bibliographystyle{IEEEtran}
\bibliography{ref}

\end{document}